\pdfoutput=1
\documentclass[11pt]{article}

\usepackage[svgnames]{xcolor}
\usepackage[margin=1in]{geometry}

\usepackage{rotating}
\usepackage{makecell}
\usepackage[numbers]{natbib}
\usepackage{booktabs}
\usepackage{multirow}

\usepackage{tcolorbox}

\newcolumntype{C}[1]{>{\centering\arraybackslash}p{#1}}

\usepackage[multiple]{footmisc}

\usepackage{float}
\usepackage{comment}
\usepackage{soul}
\usepackage{wrapfig}

\usepackage{algorithmicx}
\usepackage[noend]{algpseudocode}
\usepackage{algpseudocodex}

\usepackage{amsmath,amsthm,amssymb,mathtools}

\usepackage{graphicx}
\usepackage{dsfont}
\usepackage{tikz}
\usetikzlibrary{calc}
\usetikzlibrary{shapes.geometric}
\usepackage{paralist}

\usepackage{wrapfig}

\usepackage{float}
\usepackage{comment}
\usepackage{algorithm}
\usepackage[noend]{algpseudocode}

\usepackage{thmtools}
\usepackage{thm-restate}

\usepackage{caption}
\usepackage[labelformat=simple]{subcaption}

\usepackage{aliascnt}

\usepackage{xcolor}

\usepackage[backref=page]{hyperref}
\hypersetup{colorlinks=true,citecolor=blue,linkcolor=blue}
\hypersetup{colorlinks=true}
\hypersetup{linkcolor=[rgb]{.7,0,0}}
\hypersetup{citecolor=[rgb]{0,.7,0}}
\hypersetup{urlcolor=[rgb]{.7,0,.7}}

\declaretheorem[numberwithin=section]{theorem}
\declaretheorem[sibling=theorem, style=definition]{definition}
\declaretheorem[sibling=theorem]{lemma}
\declaretheorem[sibling=theorem]{claim}

\declaretheorem[sibling=theorem]{corollary}
\declaretheorem[sibling=theorem, style=definition]{remark}
\declaretheorem[sibling=theorem]{proposition}
\declaretheorem[sibling=theorem, name={(Failed) Proposition}]{false-proposition}

\makeatletter
\patchcmd{\ALG@step}{\addtocounter{ALG@line}{1}}{\refstepcounter{ALG@line}}{}{}
\newcommand{\ALG@lineautorefname}{Line}
\makeatother

\newcommand{\R}{\mathbb{R}}

\newcommand{\poly}{\operatorname{poly}}

\newcommand{\E}[1]{\mathds{E}\left[{#1}\right]}

\makeatletter
\newcommand\Ps@textstyle[2]{\mathbb{P}_{#1}\left[{#2}\right]}
\newcommand\Es@textstyle[2]{\mathbb{E}_{#1}\left[{#2}\right]}
\newcommand\Ps[2]{%
  \mathchoice %
  {\underset{{#1}}{\mathbb{P}}\left[{#2}\right]}
  {\Ps@textstyle{#1}{#2}}
  {\Ps@textstyle{#1}{#2}}
  {\Ps@textstyle{#1}{#2}}
}
\newcommand\Es[2]{%
  \mathchoice %
  {\underset{{#1}}{\mathbb{E}}\left[{#2}\right]}
  {\Es@textstyle{#1}{#2}}{\Es@textstyle{#1}{#2}}{\Es@textstyle{#1}{#2}}
}
\makeatother

\newcommand{\notewithcolorinitials}[3]{\colorbox{#1}{\color{white}#2:}{\textcolor{#1}{\ #3}}}

\algdef{SE}[DOWHILE]{Do}{doWhile}{\algorithmicdo}[1]{\algorithmicwhile\ #1}%

\usepackage{comment}

\usepackage{sidecap}

\usepackage{csquotes}
\usepackage{colortbl}
\usepackage{xfrac}
\usepackage{sansmath}

\usepackage{tikz}
\usetikzlibrary{automata, positioning, calc, fit, shapes.misc}
\usetikzlibrary{matrix}
\usetikzlibrary{positioning}
\usetikzlibrary{arrows, decorations.pathmorphing}
\usetikzlibrary{decorations.pathreplacing,calligraphy}

\newcommand{\Menu}{\mathsf{Menu}}

\newcommand{\Alloc}{\Sigma}

\newcommand{\V}{{\mathcal{V}}}

\newcommand{\upr}{{\mathrm{upper}}}
\newcommand{\lwr}{{\mathrm{lower}}}

\newcommand{\tax}{{\mathrm{tax}}}
\newcommand{\price}{{\mathrm{price}}}
\newcommand{\up}{{\mathrm{up}}}
\newcommand{\low}{{\mathrm{low}}}

\newcommand{\vb}{\boldsymbol}

\usepackage{comment}

\newcommand{\SubAdd}{\mathsf{SubAdd}}
\newcommand{\XOS}{\mathsf{XOS}}

\newcommand{\SingleM}{\mathsf{SingleM}}

\newcommand{\cc}{\operatorname{cc}}

\newcommand{\vectr}[1]{\mathbf{{#1}}}

\usepackage{nicefrac}

\usepackage{enumitem}

\usepackage{soul}
\usepackage[normalem]{ulem}
\newcommand{\mattnote}[1]{\notewithcolorinitials{blue}{MW}{}}
\newcommand{\ctnote}[1]{\notewithcolorinitials{orange}{CT}{}}
\newcommand{\srnote}[1]{\notewithcolorinitials{purple}{SR}{}}
\newcommand{\todonote}
[1]{\notewithcolorinitials{gray}{TODO}{}}
\newcommand{\qznote}[1]{\notewithcolorinitials{teal}{QZ}{}}

\colorlet{purple}{purple!120!}

\newcommand{\email}[1]{\href{mailto:#1}{#1}}

\title{Communication Separations for Truthful Auctions: Breaking the Two-Player Barrier}
\date{}
\author{
  Shiri Ron\thanks{Weizmann Institute of Science. Email: \email{shiriron@weizmann.ac.il}.}
  \and
  Clayton Thomas\thanks{Microsoft Research. E-mail: \email{clathomas@microsoft.com}.}
  \and
  S. Matthew Weinberg\thanks{Princeton University. E-mail: \email{smweinberg@princeton.edu}. Supported by NSF CCF-1955205.}
  \and
  Qianfan Zhang\thanks{Princeton University. E-mail: \email{qianfan@princeton.edu}. Supported by NSF CCF-1955205.}
}

\begin{document}

\maketitle

\begin{abstract}

We study the communication complexity of truthful combinatorial auctions, and in particular the case where valuations are either subadditive or single-minded, which we denote with $\SubAdd\cup\SingleM$.
We show that for three bidders with valuations in $\SubAdd\cup\SingleM$, any deterministic truthful mechanism that achieves at least a $0.366$-approximation requires $\exp(m)$ communication. In contrast, a natural extension of~\cite{Fei09} yields a non-truthful $\poly(m)$-communication protocol that achieves a $\frac{1}{2}$-approximation, demonstrating a gap between the power of truthful mechanisms and non-truthful protocols for this problem.

Our approach follows the taxation complexity framework laid out in \cite{D16b}, but applies this framework in a setting not encompassed by the techniques used in past work. In particular, the only successful prior application of this framework uses a reduction to simultaneous protocols which only applies for \emph{two} bidders \cite{AKSW20}, whereas our three-player lower bounds are stronger than what can possibly arise from a two-player construction (since a trivial truthful auction guarantees a $\frac{1}{2}$-approximation for two players).

\end{abstract}

\section{Introduction}
\label{sec:intro}
\paragraph{Background: Combinatorial Auctions.} 
Combinatorial Auctions are a paradigmatic problem at the intersection of Economics and Computation; see e.g. ~\cite{LehmannOS02,LLN01-journal,NS06,MirrokniSV08,DNS05-journal,DSS15,AssadiKS21} and many others. Here, an auctioneer has $m$ items to allocate among $n$ bidders, where each bidder $i$ has a combinatorial valuation $v_i$ over subsets of items (that is, $v_i: 2^{[m]}\rightarrow \mathbb{R}_{\geq 0}$), and seeks to do so in a way that maximizes the \emph{welfare}. That is, the auctioneer seeks to partition the items into $n$ subsets, $A_1,\ldots, A_n$, so as to maximize $\sum_i v_i(A_i)$. The challenge is that only Bidder $i$ knows her valuation $v_i(\cdot)$, and therefore: (a) the auctioneer must communicate with the bidders to learn sufficient information about $v_i(\cdot)$ to find a high-welfare allocation, and (b) the bidders may strategically manipulate the protocol, and therefore the auctioneer must further design a \emph{truthful} communication protocol.\footnote{See~\autoref{sec:prelims} for a formal definition -- informally, a protocol is truthful if every player is incentivized to follow the protocol.}

Subject only to constraint (a), the problem is fairly well understood for several canonical valuation classes.
For example, $\poly(m)$-communication (not necessarily truthful) protocols are known to achieve:
an asymptotically tight $\Theta(\nicefrac{1}{\sqrt{m}})$-approximation for arbitrary monotone valuations \cite{LS05,NS06},
a tight $\nicefrac{1}{2}$-approximation for subadditive valuations~\cite{Fei09}, a tight $(1-\nicefrac{1}{e})$-approximation is known for XOS valuations~\cite{Fei09, DNS05-journal}, and the tight constant for submodular valuations is known to lie in $[1-\nicefrac{1}{e}+\nicefrac{1}{10^{5}}, 1-\nicefrac{1}{2e}]$ \cite{DV13,FeigeV10,vondrak-thesis}.
Thus, for all of these problems, the optimal approximation guarantee of poly-communication non-truthful protocols are well-understood.

Subject only to constraint (b), the problem can be solved optimally by the classical Vickrey-Clark-Groves (VCG) mechanism~\cite{Vickrey61, Clarke71, Groves73}; however, this solution requires exponential communication for any of the above-mentioned valuation classes.\footnote{The VCG mechanism reduces non-truthful protocols to truthful ones, but only for protocols that \emph{exactly} optimize welfare. 
Thus, the VCG mechanism makes the problem easy in any setting where describing a valuation function only takes polynomially-many bits. 
Richer classes, like the four mentioned, require exponentially-many bits to describe, and require high communication for exact welfare maximization, thus motivating approximations.}
Subject to \emph{both} constraints (a) and (b), the problem is still quite poorly understood despite significant effort over the past two decades.
For example, the state-of-the-art deterministic truthful mechanisms with $\poly(m)$-communication achieve approximation guarantees of $\Omega(\nicefrac{\ln(m)}{m})$ for arbitrary monotone valuations, and just $\Omega(\sqrt{\nicefrac{\ln(m)}{m}})$ for each of subadditive/XOS/submodular valuations~\cite{QiuW24}.
Thus, these canonical settings each have a $\widetilde{\Theta}(\sqrt{m})$ gap in approximation guarantees for state-of-the-art truthful vs.~non-truthful protocols. 
Despite these massive gaps, corresponding impossibility results remain starkly rare; indeed, for each of the above examples, no lower bounds for truthful mechanisms are known beyond those that also hold for non-truthful protocols.
Thus, for each of these canonical settings, \emph{it remains unknown whether there should be any gap at all}! 

This question has received significant attention since~\cite{NS06} introduced the study of communication complexity for combinatorial auctions, and it is generally conjectured that the more significant missing piece is stronger impossibility results. Developing such impossibility results necessarily requires leveraging truthfulness, which imposes significant technical barriers. For example, a natural first attempt would be to characterize all truthful auctions in a given setting, perhaps a l\`{a} Roberts' Theorem~\cite{Rob79}, and then prove impossibility results for communication-efficient auctions using this classification.
\cite{LMN03,DN11-journal} push the classification approach roughly to its limit, but bizarre deterministic truthful mechanisms exist and it appears that any classification attempt may be intractable (see  the discussion in \cite{D11}, for example).

Interest in the question revitalized when~\cite{D16b} proposed an alternative framework termed \emph{Taxation Complexity}. 
This framework does not attempt to classify truthful mechanisms, 
but merely identifies one key complexity measure that bounds the communication complexity of a truthful auction: namely,
the logarithm of the maximum number of menus presented to a player, also called the taxation complexity. 
The Taxation Complexity framework further implies the following corollary: communication lower bounds on \emph{two-player simultaneous (non-truthful) protocols} imply communication lower bounds on \emph{two-player interactive truthful mechanisms} (whereas simultaneous lower bounds certainly do not generally imply interactive lower bounds for non-truthful protocols~\cite{Y79}). 
Even this corollary proved challenging to leverage, but~\cite{AKSW20} later established an impossibility of $\nicefrac{3}{4}-\nicefrac{1}{240}$ for two XOS bidders, which exhibits the first separation of truthful and non-truthful protocols, since there is a non-truthful protocol that gives a 
 $\nicefrac{3}{4}$-approximation for two XOS bidders~\cite{Fei09, DNS05-journal}.

\cite{AKSW20} remains to-date the only known separation between what is achievable by $\poly(m)$-communication truthful mechanisms and $\poly(m)$-communication (non-truthful) protocols. Yet, this separation holds only for two bidders. Indeed, while a $\poly(m)$-communication non-truthful protocol can guarantee a $\nicefrac{3}{4}$-approximation for two XOS bidders, the best possible guarantee for even three XOS bidders is $\nicefrac{18}{27} < \nicefrac{3}{4}-\nicefrac{1}{240}$.\footnote{
    Moreover,~\cite{BMW18} design a simultaneous protocol for two XOS bidders achieving an approximation guarantee of $\nicefrac{23}{32} > \nicefrac{18}{27}$, so no lower bound via two-player simultaneous protocols can possibly achieve a separation for more than $2$ players.}
That is, while we now know that two-bidder $\poly(m)$-communication protocols achieve strictly better guarantees than two-bidder $\poly(m)$-communication truthful mechanisms, it is still plausible that truthful mechanisms are just as powerful as non-truthful protocols for $n \geq 3$ XOS bidders. Moreover, the approach of \cite{AKSW20} heavily leverages the two-bidder corollary of~\cite{D16b}: their result follows from a communication lower bound on simultaneous protocols, which does not imply anything about interactive truthful mechanisms for $n \geq 3$ bidders.

\paragraph{Our Contributions.}
In our main result, we provide the first separation of $\poly(m)$-communication truthful vs.~non-truthful protocols for more than $2$ bidders. Specifically, we consider the class of bidders that are either Single-Minded or Subadditive (and term the class $\SubAdd \cup \SingleM$). We show that for this class, any deterministic truthful mechanism beating a $\nicefrac{(\sqrt{3}-1)}{2}\approx 0.366$ approximation requires $\exp(m)$ communication, whereas a simple extension of~\cite{Fei09} achieves a $\nicefrac{1}{2}$-approximation with a non-truthful protocol.\footnote{
  We also include novel results on related classes of valuations with two bidders; see \autoref{sec:roadmap}.
}

\paragraph{Brief Overview of Approach.}
We leverage the Taxation Complexity framework of~\cite{D16b}, which hinges around the concept of a \emph{menu}. 
The menu for Bidder $i$ of an auction $\mathcal A$ is a function $M = \Menu_i^{\mathcal{A}}(v_{-i})$, parameterized by $v_{-i}$, that takes as input $v_i$ and outputs the set of items (and price paid) that Bidder $i$ gets on input $(v_i, v_{-i})$. 
That is, a menu fixes the valuations of bidders other than $i$, and stores the impact that Bidder $i$'s valuation $v_i$ has on $i$'s own allocation and prices. 

Informally speaking, one can think of the results of \cite{D16b} as showing that, in terms of the communication-efficiency of truthful auctions (in sufficiently rich settings), the auction {might as well be implemented by fully learning a bidder's menu}, and then allocating to that bidder according to their menu and their valuation.
Indeed, in these settings, \cite{D16b} proves that for any truthful auction, the communication complexity of the auction is at least (some polynomial function of) the communication complexity of learning one bidder's menu.\footnote{
  \cite{D16b} provides other ways to bound the communication complexity of truthful auctions, as we discuss below.
}
To prove our main result, we establish that in order for a truthful mechanism to beat a $0.366$-approximation for $\SubAdd \cup \SingleM$,
the Communication Complexity of learning a bidder's menu must be exponential in the number of items $m$.
The Communication Complexity of all truthful mechanisms achieving the same approximation ratio then follows via results from~\cite{D16b}. 

We defer technical details, but now briefly give intuition for the role of both Subadditive bidders and the Single-Minded bidder, and how the $0.366$-approximation factor arises. A construction of~\cite{EFNTW19} establishes that exponential communication is necessary in order for a (not necessarily truthful) protocol to beat a $\nicefrac{1}{2}$-approximation for two subadditive bidders. So intuitively, truthful welfare-maximizing auctions face the following challenge.
Suppose there is one Single-Minded and two Subadditive bidders.
Imagine that the maximum possible welfare between the two Subadditive bidders is $c$. Then if, based on the two Subadditive bidders, the price of set $S$ for the Single-Minded bidder is set to $p(S)$, it could be that:
\begin{itemize}
    \item Perhaps the Single-Minded bidder has interest set $S$, has value barely exceeding $p(S)$ (and therefore chooses to purchase set $S$), and yet the maximum welfare achievable between the two Subadditive bidders for $\overline{S}$ is $\approx 0$. Therefore, we could have welfare as bad as $\approx p(S)$ when the optimum is $c$.
    \item Perhaps the Single-Minded bidder has interest set $S$, has value barely below $p(S)$ (and therefore chooses not to purchase anything), and yet the maximum welfare achievable between the two Subadditive bidders for $\overline{S}$ is also $c$. Moreover, even when allocating all items to the two Subadditive bidders, we don't expect to achieve welfare greater than $c/2$ without $\exp(m)$-communication, based on the results of \cite{EFNTW19}. Therefore, we could have welfare as bad as $c/2$ when the optimum is $c+p(S)$.
\end{itemize}
Therefore, the best ratio we can hope to achieve is $\min\{\frac{p(S)}{c}, \frac{c/2}{c+p(S)}\}$, which (it turns out) is at most $\frac{\sqrt{3}-1}{2}\approx 0.366$ no matter how we set $p(S)$. 

Of course, this is just intuition---we need an actual construction where the above arguments hold even after the bidders have engaged in polynomial communication. For example, if we know the Single-Minded Bidder is interested in a particular set $S$, then it is trivial to determine whether the maximum welfare achievable between the two Subadditive bidders for $\overline{S}$ is $\approx 0$ or $\approx c$. %
Nevertheless, the high-level intuition of our proof closely follows the outline above.
In particular, the role of the two Subadditive Bidders in the is to: (a) have uncertainty regarding whether $\overline{S}$ can generate non-trivial welfare between them, and (b) require $\exp(m)$ communication to beat welfare $c/2$ allocating any set of items to them. The role of the Single-Minded Bidder is for a clean and direct analysis of what the price of a single set $S$ might mean for the resulting allocation, thus allowing us to reason about the menu presented to the Single-Minded Bidder.

\subsection{Roadmap and Conclusions}
\label{sec:roadmap}

We give preliminaries in \autoref{sec:prelims}. 
In \autoref{sec:proof-prelim}, we give a ``warmup'' to our main result, namely, we give a false proof attempting to implement the intuitive hardness outline discussed above, but leaving a gap in one step of the proof.
After explaining this gap and the main ideas used to fix it, in \autoref{sec:real-proof}, we prove our main result: a welfare approximation lower bound of $0.366$ for poly-communication truthful auctions for three bidders with valuations in $\SubAdd\cup \SingleM$.
We then observe that this gives a separation between truthful auctions and non-truthful protocols, by showing that non-truthful protocols can achieve a higher welfare approximation of $\frac{1}{2}$.\footnote{
    Additionally, we note that our 
    results get a separation for any $n=O(\log n)$ bidders (since our upper bound of a $\nicefrac{1}{2}$-approximation extends to any $n=O(\log n)$ bidders, and not only $3$; see \autoref{lemma::black-box-upper-bound}).
    This is in contrast to \cite{AKSW20}: while they prove upper and lower bounds yielding a separation for $2$ bidders, there is no known separation with $3$ or more bidders in their setting (since, while their lower bounds trivially also holds for more bidders, their upper bound degrades as the number of bidders increases). 
}

In \autoref{section::two-player}, we study auctions for two bidders in the class of valuations  $\XOS\cup \SingleM$.
Observe that
to obtain a separation for two-bidder mechanisms,  it is necessary to consider a \textquote{smaller} class than $\SubAdd\cup \SingleM$.\footnote{The reason for it is that 
    for $\SubAdd \cup \SingleM$, 
    there is provably no gap between the power of communication efficient truthful mechanisms and non-truthful protocols for two bidders, as the second-price auction on the grand bundle is truthful, communication-efficient, and $1/2$-approximates the social welfare, which  is optimal even for non-truthful protocols \cite{EFNTW19}.
    Thus, to show a gap for two bidders, it is necessary to consider a \textquote{smaller} class.
} 
We show that for two bidders in the class of valuations  $\XOS\cup \SingleM$, the communication complexity of every truthful mechanism that has an approximation better than $\nicefrac{(\sqrt{5}-1)}{2} \approx 0.618$ is $\exp(m)$. We show this bound by two different proofs based on two different approaches from \cite{D16b}.
Finally, in \autoref{subsec::difficulties-xos} we also discuss difficulties for generalizing it to a stronger lower bound for more than two $\XOS\cup\SingleM$ bidders.
Holisticly, our results in \autoref{section::two-player} illustrate
additional ways to achieve impossibility results, and in fact served as stepping stones en-route to our main result.

We conclude by restating that our main result is the first separation between the approximation guarantees of $\poly(m)$-communication truthful and non-truthful combinatorial auctions beyond two bidders for any class of valuations. Still, the major open problem of whether $\poly(m)$-communication deterministic truthful mechanisms can achieve an approximation guarantee better than $\widetilde{\Omega}(\sqrt{m})$ remains open. An obvious direction for future work is to make further progress, extending our result either by removing the need for a single-minded bidder, or by achieving super-constant lower bounds for more bidders. It is also important to investigate whether our techniques open doors for similar results with other canonical valuation classes (e.g., submodular valuations, arbitrary monotone valuations), or the related setting of multi-unit auctions.

\subsection{Related Work}

\paragraph{Communication Complexity of Deterministic Truthful Combinatorial Auctions.} The most related work to ours concerns the Communication Complexity of Truthful Combinatorial Auctions. Here, the best approximation ratios are guaranteed by ``VCG-Based'' (also called ``Maximal-in-Range'') mechanisms~\cite{HolzmanKMT04,DNS05-journal,DN07a-journal,QiuW24}, which are a factor of $\widetilde{\Theta}(\sqrt{m})$ worse than those of the best non-truthful protocols~\cite{LehmannOS02, Fei09, DNS05-journal, FeigeV10} for the canonical settings of Submodular, XOS, Subadditive, and Monotone. It is plausibly conjectured that the aforementioned VCG-based mechanisms are optimal among deterministic truthful mechanisms, which would imply that strong separations between truthful and non-truthful protocols exist.

The lone prior separation is~\cite{AKSW20}.
This separation relies on a result of \cite{D16b} which reduces the problem to showing impossibilities for two-bidder simultaneous protocols; this reduction does not extend beyond two bidders.%
\footnote{
    \cite{DobzinskiNO14,DobzinskiRV22} also prove lower bounds for simultaneous combinatorial auctions. However, their results hold only for a large number of bidders.
}
In contribution to this line of works, we provide the first separation for $> 2$ bidders and directly via the Taxation Complexity framework~\cite{D16b} in a manner that is not restricted to a particular number of bidders.
Finally, we also note that the particular notion of truthfulness we consider is the standard ``Ex-Post Nash,'' meaning that it is always a Nash equilibrium for Bidders to follow the protocol, no matter their valuations. A stronger notion of truthfulness asks that it is a ``Dominant Strategy'' to follow the protocol, and~\cite{
RubinsteinSTWZ21, DobzinskiRV22} establish  separations between the achievable guarantees of $\poly(m)$-communication Dominant Strategy Truthful mechanisms and non-truthful protocols.\footnote{
    The distinction between Ex-Post Nash and Dominant-Strategy Truthful is that the former need only incentivize each bidder to behave truthfully when the other bidders are truthful \emph{according to some, but possibly arbitrary, valuation}, whereas the latter faces the much stiffer task of incentivizing each bidder to be truthful \emph{even when other bidders are behaving in a bizarre manner inconsistent with any valuation function.}
}

\paragraph{Communication Complexity of Randomized Truthful Combinatorial Auctions.} There is also a significant line of work developing randomized truthful combinatorial auctions. Here, the state-of-the-art approximation guarantees are much better, $\poly(1/\log \log m)$ for Submodular, XOS, and Subadditive~\cite{AssadiKS21} (building upon~\cite{AssadiS19, Dobzinski16a, Dobzinski07, KrystaV12}), and $\Theta(\sqrt{m})$ for arbitrary monotone valuations~\cite{DNS06-journal}. The latter is asymptotically tight, whereas it remains a major open problem whether the former can be improved to a constant.

\paragraph{Computational Complexity of Truthful Combinatorial Auctions.} A significant body of work also considers the computational complexity of truthful Combinatorial Auctions. This line of research indeed concludes strong separations between the constant-factor approximations achievable by poly-time algorithms~\cite{LLN01-journal, MirrokniSV08, Vondrak08} and poly-time truthful mechanisms~ for Submodular valuations~\cite{D11, DobzinskiV12a, DobzinskiV12b, DobzinskiV16, DughmiV11}. These results establish an $\widetilde{\Omega}(\sqrt{m})$ lower bound on the approximation guarantees of randomized poly-time truthful mechanisms. The simple posted-price mechanisms of the prior paragraph `break' these bounds because they use demand queries, and~\cite{CaiTW20} identifies a relaxed notion of ``truthfulness'' under which these mechanisms achieve their guarantees in poly-time. This counterintuitive interaction between computation and incentives further motivates the communication model for combinatorial auctions.

\section{Preliminaries}
\label{sec:prelims}
\paragraph{Combinatorial auctions.}
Recall that in a combinatorial auction, there are $m$ heterogeneous items and $n$ bidders. 
We denote items by $j \in [m] = \{1,2,\ldots,m\}$, and bidders by $i \in [n]$.
The set of allocations of items to the $n$ bidders, i.e., sets $(A_1,\ldots,A_n)$ with $A_i \subseteq [m]$ and $A_i\cap A_j = \emptyset$ for all $i \ne j$, is denoted by $\Alloc$.
Each bidder $i$ holds a private valuation function $v_i: 2^{[m]}\to \mathbb{R}$ which is drawn from some set $\V_i$. The sets $\V_i$ differ depending on the auction problem considered; we discuss different canonical cases below.
We assume that all valuations are monotone non-decreasing ($v_i(S)\le v_i(T)$ for $S \subseteq T$) and normalized ($v_i(\emptyset)=0$). 
The goal is to find an allocation of items $(A_1,\dots,A_n) \in \Alloc$ that (approximately) maximizes the social welfare $\sum_{i \in [n]} v_i(A_i)$.

An {\em auction rule}, equivalently known as a direct-revelation mechanism, is a function
$\mathcal{A} : \V_1\times \cdots\times \V_n\to \Alloc \times \R^n$, that maps a valuation profile 
$(v_1,\ldots,v_n)$ to an outcome $\big( (A_1,\ldots,A_n), p_1,\ldots,p_n \big)$, where $S_i$ is the allocation of player $i$ and $p_i$ is her payment.
The auction rule can also be equivalently be defined as a tuple $(f, p_1,\ldots,p_n)$, where the allocation is determined via the function $f: \V_1\times \cdots\times \V_n\to \Alloc$ and the payments are determined via the functions $p_i: \V_1\times \cdots\times \V_n\to \R$ for each player $i\in[n]$.
When no confusion can arise, we do not distinguish between these two representations of the auction rule.
We denote by $f_i(v_1,\dots,v_n)$ the bundle that bidder $i$ receives in $f(v_1,\dots,v_n)$.
When bidder $i$ receives bundle $S_i$ and is charged payment $p_i$, bidder $i$'s utility is $v_i(S_i)-p_i$. 
We assume that bidders are rational and strategic, and thus aim to maximize their utility.

We say that an auction rule is {\em truthful} (for the relevant classes of valuations $\V_1,\ldots,\V_n$)  if it satisfies the following:
for every bidder $i$, for every two valuations $v_i, v_i' \in \V_i$ for bidder $i$, and valuations $v_{-i}\in\V_{-i}$ for other players,
$$v_i(f_i(v_i,v_{-i}))-p_i(v_i,v_{-i}) 
\ge v_i(f_i(v_i',v_{-i}))-p_i(v_i',v_{-i}).$$
Intuitively, this means truthfully-reporting is always $i$'s best strategy.\footnote{ 
    Note that, following most of the algorithmic mechanism design literature, we consider a mechanism truthful so long as the underlying auction rule is truthful, i.e., our definition is independent of the communication protocol used to implement the mechanism.
    In technical jargon, this means we we study implementations of truthful auction rules in ex-post equilibria, i.e., \emph{ex-post} incentive-compatible auction rules, as opposed to the stronger notion of  {\em dominant-strategy} incentive compatibility (where no bidder would regret acting according the true value even if other bidders act in a way which is inconsistent with any valuation function, see e.g. \cite{RubinsteinSTWZ21, DobzinskiRV22}). 
    Note that working with a weaker notion of truthfulness makes our lower bounds only stronger.
}
An auction rule gives an \emph{$\alpha$-approximation to the optimal welfare} (for the valuation classes $\V_1,\ldots,\V_n$) if for every valuation profile $(v_1,\ldots,v_n)\in \V_1\times\cdots\times\V_n$:
$$ \sum_{i=1}^n v_i(f_i(v_1,\ldots,v_n))
 \ge \alpha \cdot \max_{\substack{(S_1,\ldots,S_n) \\ \in \Alloc}} \sum_{i=1}^n v_i(S_i).
$$

\paragraph{Communication complexity and protocols.}
We study the communication complexity of implementing auction rules (which, we recall, calculate both the allocation and the payments charged).
Technically, we study the $n$-player deterministic number-in-hand blackboard communication model, where the input known to each player $i$ is her valuation function $v_i \in \V_i$ and all the messages sent are visible to all the bidders. 
The communication cost of a protocol is the maximum number of bits that are written to the blackboard in the worst case. The communication complexity of some problem, denoted with $\cc(\cdot)$,   is the minimum communication cost of a protocol that computes it.
For clarity, we phrase our results (e.g.,) as ``for any truthful auction rule $\mathcal{A}$ for valuations $\V$ achieving an $\alpha$-approximation to the optimal welfare, the communication complexity of $\mathcal{A}$ is at least $C$''. It is well known and easy to see that the communication complexity of a protocol is at least the log of the number of transcripts; see e.g. \cite{KN97}. Our lower bounds will be based on this fact.

In our main results, we discuss three-player protocols for three-bidder auctions. 
We refer to these three bidders as Alice, Bob, and Charlie, and we denote 
$i \in \{A, B, C\}$ for shorthand
(so that, for example, these bidders' valuation functions are $v_A, v_B, v_C$).

\paragraph{Valuation classes.}
As mentioned above, different combinatorial auction problems are defined by different classes of bidder valuations; typically, one studies valuations with some canonical property.\footnote{
    Formally, specifying a class of valuation functions for use in a communication protocol requires specifying a number of items, and a range / precision of possible numerical values.
    Formally, one can say that a class of valuations uses \emph{precision $k$} if for all $v$ in the class and all bundles $S$, we have $v(S) \in \{0\} \cup \{ x / 2^k \hspace{0.25em} | \hspace{0.25em} x\in [2^{2k}] \}$.
    Following most of the algorithmic mechanism design literature, we leave $k$ implicit in all our communication complexity bounds.
    Formally, this means that we always hide factors of $\poly(k)$ in the communication costs of the protocols we construct, and that all of our lower bounds / impossibility theorems hold for some $k = \poly(m, n)$. For more discussion of the issue of representation of numbers, see  \cite[Appendix A.4.1]{D16b} and our \autoref{app:uniform}.
    \label{fn:precision-preliminaries}
}
Our main result concerns the following two classes. 

\begin{definition}[Single-minded valuations]
    A valuation function $v:2^{[m]}\to \mathbb{R}$ is {\em single-minded} if there exists a weight $w\in \R_{\ge 0}$ and a set $T \subseteq [m]$ such that $v(S)=w \cdot \mathrm{I}[S \supseteq T]$ for all $S \in 2^{[m]}$, where $\mathrm{I}[\cdot]$ denotes the indicator function. We denote the family of all single-minded valuation functions over $m$ items by $\SingleM = \SingleM_m$. 
\end{definition}

\begin{definition}[Subadditive valuations]
    A valuation function $v:2^{[m]}\to \mathbb{R}$ is {\em subadditive} if, for all bundles $S, T \subseteq [m]$, we have $v(S\cup T) \le v(S) + v(T)$. We denote the family of all subadditive valuation functions over $m$ items by $\SubAdd = \SubAdd_m$. 
\end{definition}

$\SingleM$ is in some sense the most basic class of valuations with economic complementarities---the bidder gets positive utility $w$ if and only if she gets her desired bundle $T$.
On the other hand, $\SubAdd$ is often considered the most general canonical class of valuations \emph{without} complementarities---when two bundles are combined, the bidder's utility can never increase beyond the sum of her value on the individual bundles.

\paragraph{Black-box welfare approximation reduction with $O(\log m)$ single minded bidders.}
In our paper, we focus on showing a separation between the power of communication-efficient truthful auctions and their non-truthful counterparts for the class $\SubAdd\cup \SingleM$. However, the known non-truthful protocols are defined only for subadditive bidders. Fortunately, we can easily extend the existent protocols to take care of single-minded bidders without loss in the approximation guarantee, via the following black-box reduction:
\begin{lemma}\label{lemma::black-box-upper-bound}
    Let $\mathcal P$ be a protocol that achieves an approximation $\alpha$ for $n$ bidders with valuations in the class $\mathcal C$ with $\poly(m)$ bits. Then, there exists a protocol $\mathcal P'$ that achieves an approximation $\alpha$ for $n$ bidders with valuations in the class $\mathcal C\cup \SingleM$ with $\poly(m,2^n)$ bits.
\end{lemma}
Briefly, the protocol $\mathcal{P}'$ proceeds by asking each bidder whether they are single-minded, and running protocol $\mathcal{P}$ separately for each of the $\le 2^n$ sub-problem assuming that a given subset of single-minded bidders are allocated.
We defer the full proof of \autoref{lemma::black-box-upper-bound} to \autoref{app:missing-proofs}.

\subsection{Menus}

Following the seminal work of~\cite{D16b}, our approach to studying the communication complexity of truthful auctions centers around the notion of bidders' \emph{menus}.
Given a truthful mechanism $\mathcal A$, bidder $i$'s \emph{menu} given the valuations $v_{-i}$ of the other players specifies a price $p_S$ for every subset of items $S\subseteq[m]$. 
Thus, if player $i$ wins $S$, then she pays $p_S$;
the fact that such a $p_S$ is well-defined follows immediately from truthfulness, via \autoref{prop-taxation-principle} below.
We use the following notation for menus:

\begin{definition}[Menus]
For any truthful auction rule $\mathcal{A} = (f, p_1,\ldots,p_n)$, any player $i$, and any valuation profile $v_{-i}\in\V_{-i}$, define the \emph{menu} of player $i$ as the function $\Menu_i^{\mathcal{A}}(v_{-i}): 2^{[m]} \to \R$ such that, for all bundles $S \subseteq {[m]}$ such that there exists a $v_i$ with $f_i(v_i,v_{-i})=S$,
we have $\Menu_i^{\mathcal A}(v_{-i})(S) = p_i(v_i, v_{-i})$. We often write $\Menu(v_{-i})$ where the mechanism and the player presented with the menu are clear from context. 
\end{definition}

Note that the menu is well-defined only if for every player $i$, and for every valuation profile $v_{-i}\in \V_{-i}$ of the other players, 
the price a bidder is charged can depend only on the bundle $S$ that player $i$ receives, and cannot vary with valuation $v_i$ (so long as player $i$ still receives bundle $S$). Fortunately, this fact follows immediately from truthfulness,
goes back to at least \cite{Hammond79}, and (following \cite{G81}) is known as the taxation principle:

\begin{proposition}[Taxation Principle]\label{prop-taxation-principle}
Consider a truthful auction rule $\mathcal{A} = (f, p_1,\ldots,p_n):\mathcal V_1\times\cdots\times \mathcal V_n \to \Sigma\times \mathbb R^n$.
For every bidder $i \in [n]$, every valuation profile $v_{-i} \in \V_{-i}$, and every bundle $S\subseteq [m]$, if $f_i(v_i,v_{-i})=f_i(v_i',v_{-i})=S$ for two valuations $v_i,v_i'\in \mathcal V_i$, then $p_i(v_i,v_{-i})=p_i(v_i',v_{-i})$.
\end{proposition}

Note also that truthfulness directly implies that for all $(v_i, v_{-i})$, the bundle of items $i$ receives is (one of) the bundles $S$ that maximizes $v_i(S) - M(S)$, where $M = \Menu_i(v_{-i})$.

A property of menus that will be useful is that we can assume that they are be non-decreasing and normalized without loss of generality, just like valuation functions. 
\begin{proposition}[Menu monotonicity~\cite{D16b}]\label{prop-menu-mono}
Let $\mathcal{A}=(f,p_1,\ldots,p_n):\V_1\times\cdots\times\V_n\to\Alloc\times\mathbb{R}^n$ be any deterministic truthful {auction rule}. There exists some other mechanism $\mathcal{A}'=(f,p_1',\ldots,p_n'):\V_1\times\cdots\times\V_n\to\Alloc\times\mathbb{R}^n$ with the same allocation function $f$, such that for every $i\in [n], v_{-i} \in \V_{-i}$, the menu $\Menu_i^{\mathcal{A}'}(v_{-i})$ for bidder $i$ is non-decreasing, i.e., $\Menu_i^{\mathcal{A}'}(v_{-i})(S) \le \Menu_i^{\mathcal{A}'}(v_{-i})(T)$ for $S \subseteq T$, and normalized, i.e., $\Menu_i^{\mathcal{A}'}(v_{-i})(\emptyset)=0$.
\end{proposition}

\subsection{Lower Bounds through Menus} \label{subsec::taxation}

\cite{D16b} gives several techniques for lower bounding the communication complexity of truthful auctions; all of these techniques hinge around the concept of the menu. 
To get some intuition as to how this works, observe that  one way to implement a truthful combinatorial auction is to select a player $i$, run an $(n-1)$-player communication protocol among the bidders other than $i$ to determine $i$'s menu $M = \Menu_i(v_{-i})$, then query $i$ to determine $i$'s highest-utility bundle according to $M$.\footnote{
    Note that this intuitive sketch ignores potentially-complex issues of tie-breaking; these considerations are handled in full detail in \cite{D16b}.
}
Informally, \cite{D16b} shows that something quite special happens for truthful auctions (in sufficiently rich domains), namely, 
truthful auctions can do \emph{essentially no better} in terms of communication-efficiency than this approach of explicitly learning bidders' menus.

The above intuition is actually formalized multiple different ways in \cite{D16b}, with a large emphasis on the Taxation Complexity, which counts (the log of) the number of distinct menus.
We recall this technique %
(and the reduction to simultaneous protocols exploited in \cite{AKSW20}) when we need them in \autoref{section::two-player}. 
To prove our main result, we use the following result from \cite{D16b}, which gives an even stronger lower bound technique than Taxation Complexity:

\begin{theorem}[Follows from \cite{D16b}]
\label{lemma:menu-reconstruction-cc}
    Consider any deterministic, truthful auction rule $\mathcal{A}$ for $m$ items and $n\ge 2$ bidders with valuation functions in $\mathcal{V}$.
    Suppose $\mathcal{V} \supseteq \SubAdd$.
    Let $\cc(\mathcal A)$ be the communication complexity of $\mathcal A$.
    Then, for any player $i \in [n]$, there exists an $(n-1)$-player protocol $\mathcal{P}$ with $\poly( \cc(\mathcal A),m,n)$ bits of communication that computes (the index of) the menu $\Menu_i^{\mathcal A}(v_{-i})$.
\end{theorem}

\autoref{lemma:menu-reconstruction-cc} is a direct corollary of \cite{D16b}'s Theorem 3.1, Theorem 2.3, and Proposition F.1. See the proof in \autoref{app:missing-proofs}.

\section{Building Up To Our Main Result} 
\label{sec:proof-prelim}
The main result of this paper is a lower bound on the communication complexity of getting good welfare in a three-bidder truthful auctions.
In particular, we will show that for the class of bidder valuations $\SubAdd\cup\SingleM$, any truthful auction for three bidders which gets $\frac{\sqrt{3}-1}{2}\approx 0.366$ approximation requires $\exp(m)$ communication.

In this section, we give exposition into the proof of this result via a simplified, but ultimately faulty, version of the proof which conveys much of the ideas (and how our real proof overcomes the faults of this expository version).

\subsection{Failed Proof Attempt}\label{subsec::failed-proof}

To give exposition into our proof, we first give a simpler ``failed attempt''. 
In this ``failed attempt'', there will be a gap in the proof (as we explain when the gap appears);  we will also gloss over techniques from previous works and claim without proof that a certain constructions exist.
In our actual proof, we will correct the gap, and also provide full details and proofs for all constructions.

\subsubsection{(Failed) proof outline}

Let $\mathcal{A}$ be a truthful auctions for $\SubAdd\cup\SingleM$ with three-bidder and $m$ items.
By \autoref{lemma:menu-reconstruction-cc}, the communication complexity of $\mathcal{A}$ is at most (a polynomial function of) the communication complexity of determining a bidder's menu.
Thus, our goal is to show that if $\mathcal{A}$ gets a good approximation to the optimal welfare, then the communication complexity of determining a bidder's menu is $\exp(m)$.

Our approach is inspired by the classical rectangle argument in communication complexity. This argument proves lower bounds on communication complexity by constructing a large set of inputs (the ``fooling set''), considering any pair of inputs in the fooling set, and showing that every protocol computing a function $f$ must use a different transcript for those two inputs. 
However, in order to make our argument work, it turns out we need to extend this approach to consider \emph{$4$-tuples} of inputs rather than pairs.

After constructing the valuations used in our hard instance, we lower bound the communication complexity of Bob and Charlie calculating Alice's menu. 
Our (failed) proof outline for this proceeds in two steps.
\begin{itemize}
    \item First, we show that, if a communication-efficient and truthful auction rule induces the same menu on a $4$-tuple of ``sufficiently different'' inputs, then the auction cannot get a good approximation to the optimal welfare.
    If such a $4$-tuple uses the same menu, we call this a ``bad'' $4$-tuple; for a precise definition, see below.
    This step uses direct arguments about our hard instance, and about the welfare obtainable under different menus and different possible valuations of Alice.

    \item Second, we consider any protocol for calculating Alice's menu which avoids all bad $4$-tuples, and show that all such protocols must have high communication.
    This step is a direct, simple lemma in communication complexity.
\end{itemize}

Combining these two steps finishes the (failed) proof.

\subsubsection{Description of the Class of Valuations}
\label{sec:simplified-construction}

We begin by defining the class of subadditive valuation functions that our simplified construction uses. 
This construction is based on exponentially many overlapping $4$-cell partitions of all of the items; say these partitions are denoted $\mathcal{G}_{i, 1} \cup \mathcal{G}_{i, 2} \cup \mathcal{G}_{i, 3} \cup \mathcal{G}_{i, 4} = [m]$ for many different {values of} $i$, say $i\in[K]$ for some large $K$.
The key property of the construction is that, for each way to select one cell of each partition, there is a valuation in this class who, for each partition $i$ simultaneously, wants \emph{only} elements of the selected cell of partition $i$. 
Intuitively, \cite{EFNTW19} construct valuations like this in order to hide an allocation with good welfare by reducing it to an equality problem---high welfare can be achieved among two bidders with valuations in this class if and only if there is some partition on which the bidders want different items.

Formally, for some family of partitions $\big\{ \{ \mathcal{G}_{i, 1} , \mathcal{G}_{i, 2} , \mathcal{G}_{i, 3} , \mathcal{G}_{i, 4} \} \big\}_{i,j}$, we define a class of valuations $\mathcal V$  such that there exists a constant $\ell$ such that for every pair of $i\in [K]$ and $j\in [4]$, there exists a valuation $v\in \mathcal V$ such that $v(\mathcal{G}_{i, j}) = \ell -1$ and $v([m] \setminus \mathcal{G}_{i, j}) = 1$.
We use strings $\vectr{b} \in [4]^K$ to denote ways to pick a cell of each partition, where $\vectr{b}[i]$ denotes the $i$th character in the string and $\vectr{b}[i] = j \in [4]$ denotes selecting cell $\mathcal{G}_{i, j}$ from partition $i$.

Formally, the class of valuation functions we consider is given by the following:

\begin{proposition}
    \label{thrm:valuation-class-first-attempt}
    Fix any constant $\ell$.
    For all sufficiently large $m$, and for some $K = \Omega(\exp(m))$, there exists:
    \begin{enumerate}
        \item a family of $4$-cell partitions $\{ \mathcal{G}_{i,j}\ |\ i \in [K], j\in [4] \}$ of $[m]$ (i.e., for all $i$, we have $\bigcup_{j\in[4]} \mathcal{G}_{i,j} = [m]$ and $\mathcal{G}_{i,j} \cap \mathcal{G}_{i,j'}$ for all $j\ne j'$), and
        \item a family of subadditive valuation functions $v_{\vectr{b}} : 2^{[m]} \to \R_{\ge 0}$ indexed by strings $\vectr{b}\in[4]^K$,     such that for all $\vectr{b}\in[4]^K$ and all $i\in[K]$:
    \begin{enumerate}
        \item $v_{\vectr{b}}( [m] ) = \ell$,
        \item $v_{\vectr{b}}( \mathcal{G}_{i, \vectr{b}[i]} ) = \ell - 1$, and
        \item $v_{\vectr{b}}( [m] \setminus \mathcal{G}_{i, \vectr{b}[i]} ) = 1$. 
    \end{enumerate}
    We name this class $\mathcal V_{\text{4-cell}}$.
     \end{enumerate}
\end{proposition}

\begin{figure}
    \centering

    \includegraphics[width=0.6\textwidth]{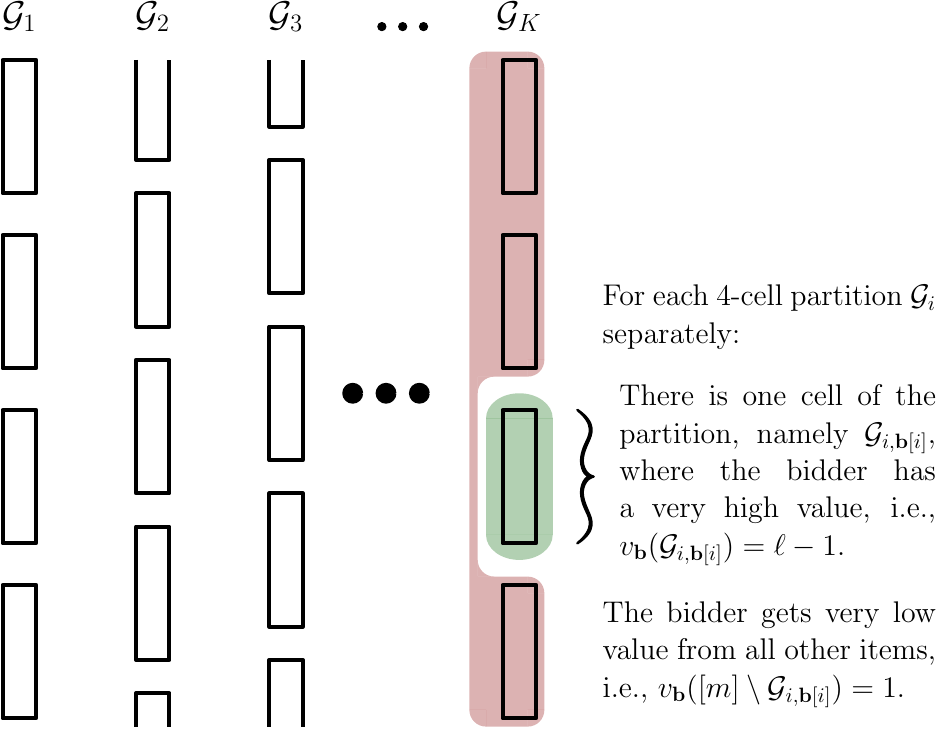}
    \caption{Illustration of valuations in \autoref{thrm:valuation-class-first-attempt}.}
        \label{fig::first-attempt}
\end{figure}

{For an illustration of the construction, see \autoref{fig::first-attempt}.} 
To prove this proposition, one can use the main construction of \cite{EFNTW19}, along with an argument using the probabilistic method.
We will later specify and formally prove an even stronger construction; for now, we assert this proposition without proof.

\cite{EFNTW19} use a construction analogous to the valuations in \autoref{thrm:valuation-class-first-attempt} to lower bound the communication complexity of maximizing welfare with two subadditive bidders.
\begin{theorem}[\cite{EFNTW19}]
    \label{thrm:efntw-hardness}
    For any communication protocol $\mathcal{P}$ that outputs an allocation that gives approximation better than $\frac{1}{2}$ for two bidders with valuations in $\mathcal{V}_{\text{4-cell}}$, the communication cost of $\mathcal{P}$ is $\exp(m)$ bits.
\end{theorem}
For deterministic protocols (our focus), one can prove this result directly by reducing to the equality problem on the strings $\vectr{b}$; for details and further results for randomized protocols (which we do not need), see \cite{EFNTW19}.
Note that \autoref{thrm:efntw-hardness} shows hardness even for non-truthful protocols, and not just for truthful auction rules.

\subsubsection{(Failed proof that) getting good welfare implies ``no bad $4$-tuple''.}

We now proceed to argue that ``sufficiently different'' elements of the valuation family given by \autoref{thrm:valuation-class-first-attempt} cannot possibly present the same menu in a truthful, communication-efficient and approximately optimal auction. 
Note, however, that this part of the proof has a gap, as we explain when this gap occurs, and later correct in the full version of the proof. 
Recall that, when we discuss three-bidder auctions, we refer to the three bidders as Alice, Bob, and Charlie (abbreviated $A, B$, and $C$).

We now describe the notion of ``bad $4$-tuples''. For some fixed auction rule, we say that a $4$-tuple $(\vectr{b}_1,\vectr{b}_2,\vectr{b}_3,\vectr{b}_4)$ of strings in $[4]^K$ is \emph{bad} if it satisfies the following two properties: 
\begin{itemize}
    \item The strings are all pairwise different on some single index $i_*$, i.e., we have $\{\vectr{b}_1[i_*], \vectr{b}_2[i_*] ,\allowbreak \vectr{b}_3[i_*] ,\allowbreak \vectr{b}_4[i_*] \} = \{1,2,3,4\}$. 
    \item Bob and Charlie present the same menu to Alice when their valuations are $(v_{\vectr b_1}, v_{\vectr b_2})$ or when they are $(v_{\vectr b_3}, v_{\vectr b_4})$, i.e., $\Menu_A(v_{\vectr b_1}, v_{\vectr b_2}) = \Menu_A(v_{\vectr b_3}, v_{\vectr b_4})$. 
\end{itemize}

We will next give a false proof that any communication-efficient auction rule achieving a good welfare approximation must avoid all bad $4$-tuples.
Intuitively, the key will be to show that a bad $4$-tuple implies that the welfare is sometimes bad, by considering Alice with a single-minded valuation such that the optimal allocation is very different on $(v_{\vectr b_1}, v_{\vectr b_2})$ versus on $(v_{\vectr b_3}, v_{\vectr b_4})$.
The gap in the proof will come from the way we apply \autoref{thrm:efntw-hardness}, as we explain below.

\begin{false-proposition}\label{prop::false}
Consider an auction rule $\mathcal A$ which is truthful for three bidders with valuations in $\SubAdd \cup \SingleM$, has a $\poly(m)$-bit communication protocol, and gets a $(\sqrt{3}-1)/2 + \xi \approx 0.366 + \xi$ approximation to the optimal welfare for some constant $\xi > 0$.
Now, suppose $\vectr{b}_1,\vectr{b}_2,\vectr{b}_3,\vectr{b}_4 \in [4]^K$ and $i_* \in [K]$ are such that $\{\vectr{b}_1[i_*], \vectr{b}_2[i_*] ,\vectr{b}_3[i_*] ,\vectr{b}_4[i_*] \} = \{1,2,3,4\}$.
Then, $\Menu^{\mathcal A}_A(v_{\vectr b_1}, v_{\vectr b_2}) \ne  \Menu^{\mathcal A}_A(v_{\vectr b_3}, v_{\vectr b_4})$.
\end{false-proposition}

\begin{proof}[(Failed) Proof.]
Fix a truthful auction rule $\mathcal A$, and four vectors  $\vectr b_1,\vectr b_2,\vectr b_3,\vectr b_4 \in [4]^K$ such that $\{\vectr{b}_1[i_*], \vectr{b}_2[i_*],\allowbreak\vectr{b}_3[i_*] ,\vectr{b}_4[i_*] \} = \{1,2,3,4\}$.
Without loss of generality, assume that $\vectr{b}_j[i^*]=j$ for $j \in [4]$.
Informally speaking, what this means is that for the ``special'' partition $\mathcal{G}_{i_*,1} \cup \mathcal{G}_{i_*,2} \cup \mathcal{G}_{i_*,3} \cup \mathcal{G}_{i_*,4} = [m]$,
we have that for $j=1,\ldots,4$, valuation $v_{\vectr b_j}$ ``loves'' set $\mathcal{G}_{i_*, j}$ and ``hates'' the rest of the items, $[m]\setminus \mathcal{G}_{i_*, j}$. 
Now, suppose for contradiction that 
$\Menu^{\mathcal A}_A(v_{\vectr b_1}, v_{\vectr b_2}) = \Menu^{\mathcal A}_A(v_{\vectr b_3}, v_{\vectr b_4})$, and denote this menu with $M$ for simplicity. 

Consider the case when Alice has a single-minded valuation $v$ such that she values the set $S=\mathcal{G}_{i_*,1} \cup \mathcal{G}_{i_*,2}$ with some value $x$ (and thus Alice values all sets containing $S$ at $x$ as well, and all sets not containing $S$ at $0$).
We will tune the precise value of $x$ later; for now, suppose only that $x < \ell$, and that $M(S) \ne x$ (which is without loss of generality, since we can always slightly perturb the value of $x$ after fixing $\mathcal A$ and $\vectr b_j$ for $j=1,\ldots,4$). 

Now, we examine two cases based on whether the mechanism $\mathcal A$ allocates $S$ to Alice or not.  
\begin{itemize}
    \item Suppose $M(S) < x$. 
    In this case, Alice must receive (a bundle containing) $S =\mathcal{G}_{i_*,1} \cup \mathcal{G}_{i_*,2}$ whenever the menu is $M$. 
    Suppose that Bob's and Charlie's valuations are $(v_{\vectr{b}_1},v_{\vectr{b}_2})$. 
    
    By the definition of $v_{\vectr{b}}$ and the fact that $\vectr{b}_1[i^*], \vectr{b}_2[i^*] \in \{1,2\}$, we have that $v_{\vectr{b}_1}([m]\setminus S), v_{\vectr{b}_2}([m]\setminus S) \le 1$.
    Thus, $\mathcal A$ gets welfare at most $x + 2$.
    
    On the other hand, consider the allocation that gives nothing to Alice, gives $\mathcal{G}_{i_*,1}$ to Bob, and gives $\mathcal{G}_{i_*,2}$ to Charlie. 
    The welfare of this allocation is $2(\ell-1)$.
    Thus, the welfare approximation ratio of $\mathcal A$ in this case is at most $(x + 2) / \big(2(\ell-1)\big)$.
    
    \item Suppose $M(S) > x$. 
    In this case, Alice cannot receive a bundle containing $S=\mathcal{G}_{i_*,1}\cup\mathcal{G}_{i_*,2}$ when the menu is $M$.
    Suppose that Bob's and Charlie's valuations are $(v_{\vectr{b}_3},v_{\vectr{b}_4})$. 
    
    \textbf{False Argument:}
    Now, one may \emph{intuitively} think that since Alice gets no items in this case, we are just solving the allocation problem of awarding items between Bob and Charlie, who are two bidders with valuations drawn from the class of valuations specified in \autoref{thrm:valuation-class-first-attempt}.
    On these instances, giving $\mathcal{G}_{i_*,3}$ to Bob and $\mathcal{G}_{i_*,4}$ to Charlie gets welfare $2(\ell-1)$.
    However, \emph{intuitively}, by \autoref{thrm:efntw-hardness}, any communication-efficient auction rule should sometimes get welfare at most $\ell$, reflecting the $(1/2)$-approximation hardness result of \cite{EFNTW19}.
    This argument is \emph{not actually correct}, because the communication-efficient protocol for $\mathcal A$ can actually use information from Alice's valuation to focus on the specific partition $\mathcal{G}_{i_*}$.\footnote{
        Put another way, \autoref{thrm:efntw-hardness} directly implies the following: if a poly-communication auction rule never allocates any items to Alice, then there must exist some inputs where the auction rule gets welfare $\ell$, while the optimal welfare Bob and Charlie could get amongst themselves is at least $2(\ell -1)$.
        Completing this proof would require a stronger property:
        for \emph{any} poly-communication auction rule, there exists some case where the auction rule gets welfare $\ell$, but the optimal total welfare is in fact $x + 2(\ell-1)$.
    }
    However, to illustrate the idea, let us pretend that this is true for the remainder of the proof, and that $\mathcal A$ gets welfare at most $\ell$ in this case.

    On the other hand, consider the allocation that gives $S$ to Alice, gives $\mathcal{G}_{i,3}$ to Bob, and gives $\mathcal{G}_{i,4}$ to Charlie.
    The welfare of this allocation is $x+2(\ell-1)$, and thus the optimal welfare on this instance is at least this large.
    Thus, the welfare approximation ratio in this case is at most $\ell / \big(x + 2(\ell -1) \big)$.
\end{itemize}
Therefore, under our assumption that the menus for Alice with respect to $(v_{\vectr{b}_1},v_{\vectr{b}_2})$ and $(v_{\vectr{b}_3},v_{\vectr{b}_4})$ are identical, the welfare approximation ratio of $\mathcal A$ can be at most
$$\max\left\{ \frac{x+2}{2(\ell-1)}, \frac{\ell}{x+2(\ell-1)}\right\}.$$
Taking $x=(\sqrt{3}-1)\ell$, the ratio above converges to $\frac{\sqrt{3}-1}{2} \approx .366$ as $\ell \to \infty$.\footnote{
    This specific value of $x$ arises from assuming $x = \rho \ell$ for some $\rho \in (0,1)$, then making the two terms in the above maximum equal to each other while ignoring lower-order factors.
    In more detail, assuming that $\rho\ell / (2 \ell) = \ell / ( \rho\ell + 2\ell)$ for $\rho > 0$ implies that $\rho^2 + 2\rho - 2 = 0$, and thus that $\rho = \sqrt{3} - 1$. Then, the above maximum goes to $\rho / 2$ as $\ell \to \infty$. 
}
\end{proof}

\subsubsection{``No bad $4$-tuple'' implies high communication.}

\autoref{prop::false} implies in particular that if the auction rule has any bad $4$-tuples, then the auction cannot get a good approximation to the optimal welfare.
We will show next that for any protocol which avoids all bad $4$-tuples, the communication cost of the protocol must be high.
For clarity and generality, we state this part of the argument for general communication problems.

\begin{lemma}[$4$-tuple rectangle argument]
\label{lem:4-tupe-rectangle-simple}
    For a function $f : \mathcal{T} \times \mathcal{T} \to \mathcal{R}$, let $\mathcal{S} = \{ v_{\vectr{b}} \}_{\vectr{b}\in [4]^K} \subseteq \mathcal{T}$ be a set of distinct inputs to $f$ parameterized by a $K$-length string $\vectr{b} \in [4]^K$.
    Suppose that for every 4-tuple $(\vectr{b}_1,\vectr{b}_2,\vectr{b}_3,\vectr{b}_4) \in \mathcal{S}^4$ that differs on some index $i \in [K]$ (i.e., $\{\vectr{b}_1[i],\vectr{b}_2[i],\vectr{b}_3[i],\vectr{b}_4[i]\} = \{1,2,3,4\}$),
    we have $f(v_{\vb{b}_1},v_{\vb{b}_2}) \ne f(v_{\vb{b}_3},v_{\vb{b}_4})$.
    Then, the communication complexity of $f$ is at least \[ \log_2\left(4^K / 3^K \right) = \Omega(K). \]
\end{lemma}
\begin{proof}
  Consider any communication protocol for $f$, and for any inputs $v_B, v_C \in \mathcal{T}$, let $\tau(v_B, v_C)$ be the transcript of the protocol on inputs $v_B, v_C$.
  Recall that, as in standard rectangle arguments, if $f(v_B, v_C) \ne f(v_B', v_C') $ then $\tau(v_B, v_C) \ne \tau(v_B', v_C') $. Also, if $\tau(v_B, v_C) = \tau(v_B', v_C')$ then $\tau(v_B', v_C) = \tau(v_B, v_C') = \tau(v_B, v_C)$.

  To show our lower bound, we focus on the transcripts used when both players have the same input in $\mathcal{S}$.
  For any transcript $T$ of the protocol, let $\mathcal{I}(T)\subseteq \mathcal{S}$ denote the (possibly empty) set of all $v_{\vectr{b}}\in\mathcal{S}$ such that the transcript given the input $(v_{\vectr{b}}, v_{\vectr{b}})$ is $T$. 
  
We will now show that for all transcripts $T$, we have $|\mathcal{I}(T)|\le 3^K$.
For that, we claim that for every transcript $T$ and for every $i\in[K]$, we have $\{ \vectr b[i] \ \big|\ \vectr b \in \mathcal{I}(T) \} \ne \{ 1,2,3,4\}$.
To see this, assume for contradiction that there is such a transcript $T$ and an index $i$ 
such that 
$\tau(v_{\vectr{b}_j}, v_{\vectr{b}_j}) = T$ for each $j=1,\ldots,4$. 
Then, by a standard rectangle argument, we have $T = \tau(v_{\vectr{b}_1},v_{\vectr{b}_2}) = \tau(v_{\vectr{b}_3},v_{\vectr{b}_4})$ as well. 
However, by the assumption of the lemma, we have $f(v_{\vectr{b}_1},v_{\vectr{b}_2}) \ne f(v_{\vectr{b}_3},v_{\vectr{b}_4})$, which is a contradiction.
Thus, since the string corresponding to elements in $\mathcal{I}(T)$ can take at most $3$ values on each index $i$, we know that $|\mathcal{I}(T)|\le 3^K$.

Now, this implies that any protocol correctly computing $f$ requires at least $|\mathcal{S}| / 3^K = 4^K / 3^K$ distinct transcripts in total.
Thus, %
the communication complexity of the protocol has to be $\Omega(K)$.

\end{proof}

If the failed Proposition~\ref{prop::false} had actually been true, then combining it with \autoref{lem:4-tupe-rectangle-simple} would immediately imply that the communication complexity of the two-player function $f = \Menu_A(\cdot,\cdot)$ is $\Omega(K) = \exp(m)$.
Furthermore, by the techniques of \cite{D16b} recalled in our \autoref{lemma:menu-reconstruction-cc}, we could then conclude that the auction rule itself has communication complexity at least $\exp(m)$.

\subsection{Correcting the gap in this proof}

We now briefly discuss the main idea we use below to correct the gap in the failed proof of Proposition~\ref{prop::false}.  
For an argument like the above to work, we need to
have a construction such that it will not only be hard (communication-wise) to decide whether we want to allocate to Alice or not,
but even conditional on our decision regarding Alice, it will be communication-hard for us to find a \textquote{good enough} partition of the items of Bob and Charlie. 
For that, we create a two-layered construction, which we describe in detail in \autoref{sec:real-proof}.\footnote{
    Note that we do not overrule that it is possible to prove \autoref{thm:truthful-lower-bound-for-subadditive-cup-sm} using the construction specified in this failed attempt, i.e., we do not know
    whether the exists a $\poly(m)$-communication truthful auction for the classes of valuations considered above.
}
This two-layered construction is more involved that the one discussed in %
\autoref{sec:simplified-construction} and \autoref{thrm:efntw-hardness}, but its proof is a straightforward extension of the techniques used, i.e., those of \cite{EFNTW19}. %

After making this adjustment, the (false) Proposition~\ref{prop::false}
remains nearly identical in how we argue about the welfare achieved by the auction with a ``bad 4-tuple''. However, the definition of a ``bad $4$-tuple'' becomes considerably more involved.
For instance, a ``bad 4-tuple'' is defined by
a condition both about protocols for calculating the menu, and about protocols for calculating the auction rule.
This technical fact precludes using a clean and self-contained ``generalized rectangle argument'' such as \autoref{lem:4-tupe-rectangle-simple}; still, our transcript-counting arguments follow much of the intuition provided by \autoref{lem:4-tupe-rectangle-simple}.
For the full argument, see \autoref{sec:real-proof}.

\section{Main Result: A Separation for Three Bidders}
\label{sec:real-proof}

We now state and prove the main results of our paper: a separation for three-bidder truthful auction vs. non-truthful protocols. 
First, we state the lower bound for $\poly$-communication truthful auctions with bidders in $\SubAdd\cup\SingleM$:
\begin{theorem}\label{thm:truthful-lower-bound-for-subadditive-cup-sm}
    Every deterministic truthful mechanism $\mathcal{A}$ for three bidders with valuations in $\SubAdd_m\cup\SingleM_m$ that guarantees a $(\frac{\sqrt{3}-1}{2}+\frac{12}{\log m})$-approximation to the optimal social welfare requires $\exp\big({\Omega(\frac{\sqrt{m}}{\log m})}\big)$ bits of communication. 
\end{theorem}

We prove this lower bound below.
To establish a separation, we must also give an upper bound for non-truthful protocols.
We prove this next, i.e., we construct a $\poly(m)$-communication protocol getting a good welfare approximation.
This protocol is a simple extension of a result of \cite{Fei09}:
\begin{theorem}[\cite{Fei09}]\label{thm::feige}
    For any number of bidders $n$ and items $m$,
    there is a deterministic, $\poly(m,n)$-communication (non-truthful) protocol that guarantees a $\frac12$-approximation to the optimal social welfare for bidders with valuations in $\SubAdd_m$.  
\end{theorem}

Applying the reduction in \autoref{lemma::black-box-upper-bound} to the communication efficient protocol
in \autoref{thm::feige} gives the following:

\begin{corollary}\label{lemma::ub--for-subadditive-cup-sm}
For any number of bidders $n=O(\log m)$,
there is a deterministic $\poly(m)$-communication (non-truthful) protocol that guarantees a $\frac12$-approximation to the optimal social welfare for bidders with valuations in $\SubAdd_m\cup \SingleM_m$. 
\end{corollary}

Together, \autoref{thm:truthful-lower-bound-for-subadditive-cup-sm} and \autoref{lemma::ub--for-subadditive-cup-sm} give the main result of our paper: truthful mechanisms are provably less powerful than non-truthful protocols for bidders with valuations in $\SubAdd\cup \SingleM$. Observe that the upper bound, and thus the separation, holds for every constant $n \ge 3$, and indeed even for every number of bidders that is logarithmic in the number of items.\footnote{
    Observe that a trivial second-price auction on the grand bundle among the three bidders in $\SubAdd\cup \SingleM$ that gives at least $\frac{1}{3}$ of the optimal welfare, is truthful and communication-efficient, whilst by \autoref{thm:truthful-lower-bound-for-subadditive-cup-sm}, no truthful and communication efficient has approximation better than $\approx 0.366$. Hence, there remains a small gap in understanding the achievable approximation ratio for truthful mechanisms in this class.
}

We are now ready to prove \autoref{thm:truthful-lower-bound-for-subadditive-cup-sm}; this proof occupies the remainder of \autoref{sec:real-proof}.

\subsection{Proof Outline}
As discussed in \autoref{sec:proof-prelim}, our proof proceeds as follows.
To begin, we construct a class of subadditive valuations, generalizing~\cite{EFNTW19} and our construction in \autoref{thrm:valuation-class-first-attempt}.
Then, the proof of our communication lower bound proceeds in two major steps.
First, we show that getting a good welfare approximation implies that the auction must always avoid situations we term ``bad $4$-tuples'' (see below for a definition).
Second, we show that any protocol which always avoids bad $4$-tuples requires high communication.

\subsection{Description of the class of valuations}
Our construction is a significant generalization the construction of ~\cite{EFNTW19}. Our goal is to have a set of  subadditive valuations of two bidders, whom we name Bob and Charlie, and 
a collection of subsets $\mathcal G=\{G_1,\ldots,G_k\}\subseteq 2^{[m]}$ for a large enough $k$ such that for each $i\in [k]$: 
\begin{enumerate}
    \item[(1)] For every $i\in [k]$, we can find valuations of Bob and Charlie that jointly have high value for $G_i$ and low value for $\overline G_i$. Alternatively, we can find valuations of Bob and Charlie that have low value for $G_i$ and high value for $\overline G_i$.
    \item[(2)] For every $i\in [k]$ and $S\in \{G_i,\overline G_i\}$, even if we know that Bob and Charlie have high joint value for $S$, dividing the items in $S$ among Bob and Charlie in an approximately optimal way takes exponentially-many bits of communication. 
\end{enumerate}

The construction in \autoref{thrm:valuation-class-first-attempt} meets condition (1).\footnote{
    In more detail, \autoref{thrm:valuation-class-first-attempt} constructs a large set of $4$-cell partitions $\{ \{ \mathcal{G}_{i,1}, \mathcal{G}_{i,2}, \mathcal{G}_{i,3}, \mathcal{G}_{i,4} \} \ \mid \ i\in[k] \}$, and a set of valuations of Bob and Charlie, such that the set family 
    $\{ \mathcal{G}_{i,1} \cup \mathcal{G}_{i,2} \ \mid \ i \in [k] \}$ meets condition (1) above.
    Our construction extends the one in \autoref{thrm:valuation-class-first-attempt} by replacing the index $j \in [4]$ with another family of sets $\{ H_j \ \mid\ j \in [k]\}$ in order to meet condition (2).
}
For  condition (2) to hold, we build the following two-layered collection of subsets:
\begin{definition}
        A {\em $k$-width-family} $\mathcal{F}=(\mathcal{G},\{\mathcal{H}_i^{(t)}\}_{t \in \{0,1\}, i \in [k]})$ of $[m]$ consists of $2k+1$ collections $\mathcal{G},\mathcal{H}_1^{(0)},\dots,\mathcal{H}_k^{(0)}$, $\mathcal{H}_1^{(1)},\dots,\mathcal{H}_k^{(1)} \subseteq 2^{[m]}$ of size $k$ each, where we 
    denote $\mathcal{G}=\{G_1,\dots,G_k\}$ and $\mathcal{H}_i^{(t)}=\{H_{i,1}^{(t)},\dots,H_{i,k}^{(t)}\}$ for each $t \in \{0,1\},i \in [k]$ such that for all $i,j \in [k]$, $H_{i,j}^{(0)} \subseteq \overline{G_i}$ and $H_{i,j}^{(1)} \subseteq G_i$.
\end{definition}
    For every $k$-width-family $\mathcal F$, a vector $\vb{b} \in \{0,1\}^{k}$, and a $k\times k$ binary matrix $\vb C \in \{0,1\}^{k\times k}$,
    denote with $\mathcal F[\vb b,\vb C]$ the collection of $k^2$ subsets $\{S_{i,j}\}_{i,j\in [k]}$, where for each $i,j \in [k]$, $$S_{i,j} = \begin{cases}
            G_i \cup H_{i,j}^{(0)} & \vb b[i]=0, \vb C[i][j]=0, \\
            G_i \cup (\overline{G_i} \setminus H_{i,j}^{(0)}) & \vb b[i]=0, \vb C[i][j]=1,\\ \overline{G_i} \cup H_{i,j}^{(1)} & \vb b[i]=1, \vb C[i][j]=0, \\
            \overline{G_i} \cup (G_i \setminus H_{i,j}^{(1)}) & \vb b[i]=1, \vb C[i][j]=1.
                   \end{cases}$$
Recall the following property for a collection (as defined in~\cite{EFNTW19}):%
\begin{definition}
    A collection $\mathcal{S}$ of subsets in $[m]$ is called {\em $\ell$-sparse} if for all $T_1,T_2,\dots,T_{\ell-1} \in \mathcal{S}$, $\bigcup_{j=1}^{\ell-1} T_j \ne [m]$.
\end{definition}

In words, an $\ell$-sparse collection requires at least $\ell$ elements to cover all of $[m]$.
We will now use $\ell$-sparseness to define another property (mildly generalizing the corresponding property in \cite{EFNTW19}):
\begin{definition}
    A $k$-width-family is \emph{$\ell$-independent} if for all $\vb{b} \in \{0,1\}^{k},\vb C\in \{0,1\}^{k\times k}$, the collection $\mathcal F[\vb b,\vb C]$ is $\ell$-sparse. 
\end{definition} 

For exposition, we write an example of a $2$-independent $2$-width-family in \autoref{app:example}. 
By using the probabilistic method, we show that there exists such a family for {large enough} values of $k$ and $\ell$:
\begin{lemma}\label{lemma:large-family-existence}
    For sufficiently large $m$, there exists an 
    $\ell$-independent $k$-width-family, where $\ell=\frac{1}{4} \log m$ and $k=\exp\left({\frac{2\sqrt{m}}{\log m}}\right)$.
\end{lemma}
The proof of \autoref{lemma:large-family-existence} is in \autoref{app:missing-proofs}; the proof uses a (somewhat involved, but standard) application of the probabilistic method. 
In accordance with \cite{EFNTW19}, we translate $\ell$-sparse collections into modified set-cover valuations, denoted as $v_{\mathcal{S}}^{\ell}(\cdot)$. 
This translation is exactly as in \cite{EFNTW19}, but we include it for completeness.
Given an $\ell$-sparse collection $\mathcal S=\{S_1,\ldots,S_d\}$, for every $X\subseteq M$:
$$
\sigma_{\mathcal S}(X)=\begin{cases}
    \min \left\{|Y| \hspace{0.25em} \Big| \hspace{0.25em} Y\subseteq [d],\ X \subseteq \bigcup_{i\in Y}S_i \right\}  
    &\quad \text{if $X\subseteq \bigcup_{i\in [d]}S_i$} \\
    \max\{\ell,d\}  & \quad\text{otherwise.}
\end{cases}
$$
In words, $\sigma_{\mathcal S}(X)$ is the minimal number of subsets of $\mathcal S$ that cover $X$, if it exists, and otherwise it is a \textquote{large} number. 
We now use $\sigma_{\mathcal S}$ to 
define a valuation function $v_{\mathcal S}^\ell$. For every subset of items  $X\subseteq M$, if $\phi_{\mathcal S}(X)<\frac{\ell}{2}$, set $v_{\mathcal S}^\ell(X)=\sigma_{\mathcal S}(X)$ and 
$v_{\mathcal S}^\ell(\overline X)=\ell -\sigma_{\mathcal S}(X)$. For every subset that remains undefined, let $v_{\mathcal S}^\ell(\overline X)=\frac{\ell}{2}$. 
In principle, it is possible that $v_{\mathcal S}^\ell$ is not well defined. However, \cite{EFNTW19} shows that: 
\begin{lemma}[\cite{EFNTW19}]\label{lemma::valuation-properties-EFNTW19}
    For any $\ell$-sparse collection $\mathcal{S}$ and $v_{\mathcal{S}}^{\ell}(\cdot)$:
    \begin{enumerate}
            \item $v_{\mathcal{S}}^{\ell}(\cdot)$ is well defined. 
        \item $v_{\mathcal{S}}^{\ell}(\cdot)$ is monotone, normalized and subadditive.
        \item For every $S \subseteq [m]$, $v_{\mathcal{S}}^{\ell}(S) + v_{\mathcal{S}}^{\ell}(\overline{S}) = \ell$. \label{condicondi-identical}
        \item For every $S \in \mathcal{S}$, $v_{\mathcal{S}}^{\ell}(S)=1$ and $v_{\mathcal{S}}^{\ell}(\overline{S})=\ell - 1$.
    \end{enumerate}
\end{lemma}
We note one counter-intuitive aspect of the valuations defined by $\mathcal S$:  the valuation $v_{\mathcal S}^{\ell}$ has \textquote{low} value for the bundles $S$ that are in the collection $\mathcal S$.

We are now ready to define the class of valuations for which we show hardness. 
Given an
    $\ell$-independent $k$-width-family
$\mathcal{F}=(\mathcal{G},\{\mathcal{H}_i^{(t)}\}_{t \in \{0,1\}, i \in [k]})$ and denote $\mathcal{G}=\{G_1,\dots,G_k\}$ and $\mathcal{H}_i^{(t)}=\{H_{i,1}^{(t)},\dots,H_{i,k}^{(t)}\}$ for each $t \in \{0,1\},i \in [k]$, we define the following sub-class of subadditive valuations:
\begin{align*}
   \SubAdd^*&= \big\{
   v_{\mathcal{F}[\vb{b},\vb{C}]}^\ell(\cdot)\ |\ \vb{b} \in \{0,1\}^{k}, \vb{C} \in \{0,1\}^{k\times k} \big\} \subseteq \SubAdd_m.
\end{align*}
The following claim lists several properties of the valuation functions in $\SubAdd^*$, which we will use in \autoref{subsec::good-implies-bad-4-tuple}: 
\begin{claim}\label{claim::vals-props}
Fix a vector $\vb b\in \{0,1\}^k$ and a matrix $ \vb C\in \{0,1\}^{k\times k}$, and let $v(\cdot)$ be  the valuation 
parameterized by the collection $\mathcal{F}[\vb{b},\vb{C}]$. Then, for every pair of indices $i,j\in [k]$:
    \begin{enumerate}
        \item If $\vb b[i]=0$, then $v(G_i)=1$. \label{b0-case}
        \item If $\vb b[i]=0$ and $\vb C[i][j]=0$, then $v(\overline G_i \setminus H_{i,j}^{(0)})=\ell-1$. \label{b0-c0-case}
        \item If $\vb b[i]=0$ and $\vb C[i][j]=1$, then $v(H_{i,j}^{(0)})=\ell-1$. \label{b0-c1-case}
        \item If $\vb b[i]=1$ and $\vb C[i][j]=0$, then $v(G_i \setminus H_{i,j}^{(1)})=\ell-1$.\label{b1-c0-case}
        \item If $\vb b[i]=1$ and $\vb C[i][j]=1$, then $v(H_{i,j}^{(1)})=\ell-1$.\label{b1-c1-case}
    \end{enumerate}
\end{claim}
We skip the proof of \autoref{claim::vals-props}, as it directly follows from the construction and \autoref{lemma::valuation-properties-EFNTW19}. 

The class of single-minded valuations $\SingleM^* \subseteq \SingleM_m$ satisfies that each valuation is parameterized by an index   $i\in [k]$ and $\delta\in \{0,1\}$ such that:
\begin{equation} \label{eq::sm-def}
\forall X\subseteq M, \quad   v_{i,\delta}(X)=\begin{cases}
    (\sqrt{3}-1)\ell + \delta \quad &\overline G_i\subseteq X, \\
    0 \quad & \text{otherwise.}
\end{cases}
\end{equation}

\subsection{Getting good welfare implies ``no bad 4-tuple''} \label{subsec::good-implies-bad-4-tuple} 

Now, fix a truthful communication-efficient three-player mechanism $\mathcal A$ that gives an approximation strictly better than $\frac{\sqrt 3-1}{2}+\frac{3}{\ell}$ when all the bidders' valuation classes are equal to $\SubAdd_m\cup \SingleM_m$, where $\ell=\frac{1}{4} \log m$.
For convenience, we name the bidders Alice, Bob and Charlie.

The auction rule $\mathcal{A}$ defines two other communication problems which we reason about and use to define a notion of ``bad $4$-tuples''.
These two communication problems are to calculate the social welfare of $\mathcal{A}$ and to calculate the menu in $\mathcal{A}$.
Let $\mathcal{SW}(v_A,v_B,v_C)$ be the communication problem between all three bidders that computes the social welfare that $\mathcal{A}$ gets with valuations $(v_A,v_B,v_C)$, and denote the transcript function of its most efficient protocol by $\tau_{\mathcal{SW}}(v_A,v_B,v_C)$. Note that $\cc(\mathcal{SW})\le \cc(\mathcal A)+\poly(m)$.\footnote{
  This follows because one can always calculate $\mathcal{SW}$ by simply calculating $\mathcal{A}$, then having each bidder send their value for the set of items they receive (and we assume that bidders' valuations use precision $\poly(m)$, as discussed in \autoref{sec:prelims}).
}
Also, let $\mathcal{P}(v_B,v_C)$ be the communication problem between Bob and Charlie that finds the menu $\Menu_A(v_B,v_C)$ for Alice defined by the auction $\mathcal{A}$, and denote the transcript function of its most efficient protocol by $\tau_{\mathcal{P}}(v_B,v_C)$. 

We now define $4$-bad tuples and show that they cannot exist when $\mathcal{A}$ gets a good approximation to the optimal welfare. 
\begin{proposition}\label{prop::bad-tuple}
    Let $v^{(1)},v^{(2)},v^{(3)},v^{(4)}$ be valuations in $\mathcal \SubAdd^*$, parameterized by the vectors  $\vb b^{(1)},\vb b^{(2)}, \vb b^{(3)}, \vb b^{(4)}\in \{0,1\}^k$ and the matrices $\vb C^{(1)},\vb C^{(2)}, \vb C^{(3)}, \vb C^{(4)}\in \{0,1\}^{k\times k}$, respectively. 
    We say that $\{ (\vb b^{(j)}, \vb C^{(j)}) \}_{j\in [4]}$ is a \emph{bad $4$-tuple} if there exist $i^{\ast},j_1^\ast,j_2^\ast \in [k]$ such that the five following conditions hold simultaneously:  
        \begin{enumerate}
        \item $\vb b^{(1)}[i^*]=\vb b^{(2)}[i^*]=0$ and $\vb b^{(3)}[i^*]=\vb b^{(4)}[i^*]=1$. \label{bad-cond1}
        \item $\vb C^{(1)}[i^\ast][j_1^\ast]=0$ and $\vb C^{(2)}[i^\ast][j_1^\ast]=1$.  \label{bad-cond2}
                \item $\vb C^{(3)}[i^\ast][j_2^\ast]=0$ and $\vb C^{(4)}[i^\ast][j_2^\ast]=1$. \label{bad-cond3} 
        \item $\tau_{\mathcal{P}}(v^{(1)},v^{(1)})=\tau_{\mathcal{P}}(v^{(2)},v^{(2)})=\tau_{\mathcal{P}}(v^{(3)},v^{(3)})=\tau_{\mathcal{P}}(v^{(4)},v^{(4)})$, i.e. the transcript of the protocol that computes the menu is identical for all these pairs. \label{condi-4}
        \item    $\tau_{\mathcal{SW}}(v_{i^*,0}, v^{(3)}, x^{(3)})=\tau_{\mathcal{SW}}(v_{i^*,0}, v^{(4)}, x^{(4)})$,
        i.e. the transcript of the protocol that computes the social welfare of $\mathcal A$ is identical for both of these valuation profiles.\footnote{We remind that $v_{i^*,0}$ is the single-minded valuation defined in \autoref{eq::sm-def}.} \label{bad-condi5}
    \end{enumerate}
    Then, there cannot exist any bad $4$-tuples (i.e., for all such $\{ (\vb b^{(j)}, \vb C^{(j)}) \}_{j\in [4]}$, there cannot exist  $i^{\ast},j_1^\ast,j_2^\ast \in [k]$ such that all five above conditions hold).
\end{proposition}
Informally, this means that for every 
\textquote{sufficiently different} elements of the valuation family $\SubAdd^*$, their transcript either in the protocol of menu computation of $\mathcal A$ or the protocol that computes the welfare of $\mathcal A$ should differ.    
\begin{proof}[Proof of \autoref{prop::bad-tuple}]
We will show that if all conditions hold simultaneously for some indices $i^\ast,j_1^\ast,j_2^\ast$, then the approximation ratio must be strictly smaller than $\frac{\sqrt 3-1}{2}+\frac{3}{\ell}$, which leads to a contradiction.  We do that by showing that in particular, the mechanism outputs a bad approximation for one of the valuation profiles $(v_{i^*,1}, v^{(1)},v^{(2)})$  or
$(v_{i^*,0}, v^{(3)},v^{(4)})$. 
Note that by condition~\ref{condi-4} and by the rectangle argument, $\tau_{\mathcal{P}}(v^{(1)},v^{(2)})=\tau_{\mathcal{P}}(v^{(3)},v^{(4)})$, so in particular 
the menus presented to Alice given the valuations $(v^{(1)},v^{(2)})$ and $(v^{(3)},v^{(4)})$ are identical. 
Denote this menu with $M_A=\Menu_A^\mathcal{A}(v^{(1)},v^{(2)})=\Menu_A^\mathcal{A}(v^{(3)},v^{(4)})$ and recall that $M_A(S)$ is the price in the menu for any bundle $S \subseteq [m]$.
Note that by \autoref{prop-menu-mono}, we can assume that all prices are non-negative.
We complete the proof by case analysis. 

Note that if $M_A(\overline {G}_{i^\ast}) 
\le (\sqrt{3}-1)\ell$, then given the valuation profile $(v_{i^*,1},v^{(1)}, v^{(2)})$, Alice wins $\overline {G}_{i^\ast}$ because all prices are non-negative, so this is the only valuable bundle for her. 
However, since both valuations $v^{(1)},v^{(2)}$ are parameterized by $\vb b^{(1)}, \vb b^{(2)}$ such that $\vb b^{(1)}[i^\ast]= \vb b^{(2)}[i^\ast]$, by Property~\ref{b0-case} in \autoref{claim::vals-props}, we have that $v^{(1)}(G_{i^\ast})=v^{(2)}(G_{i^\ast})=1$. As a result, the total welfare of the allocation of $\mathcal A$ is at most $((\sqrt{3}-1)\ell +1)+2$. However, the optimal welfare is at least $2(\ell - 1)$, obtained by giving $\overline{G}_{i^*} \setminus H_{i^*,j^*_1}^{(0)}$ to Bob and $H_{i^*,j^*_1}^{(0)}$ to Charlie (according to Property~\ref{b0-c0-case}~and~\ref{b0-c1-case} of \autoref{claim::vals-props}).

However, if $M_A(\overline{G}_{i^*})>(\sqrt{3}-1)\ell$, we claim it still implies a failure on another valuation profile $(v_{i^*,0}, v^{(3)},v^{(4)})$. 
For that, we analyze the output of the mechanism $\mathcal A$ on $(v_{i^*,0}, v^{(3)},v^{(3)})$. 

Since the price of $\overline G_{i^\ast}$ that Bob and Charlie present to Alice is 
 is $M_A(\overline{G}_{i^*})$ given 
$(v^{(3)},v^{(3)})$, she does not win $\overline G_{i^\ast}$, resulting a welfare of zero. 
In addition, since Bob and Charlie's valuations are identical, the total welfare from them is at most $\ell$ by Condition~\ref{condicondi-identical}~of~\autoref{lemma::valuation-properties-EFNTW19}. 
Therefore, the social welfare that $\mathcal A$ obtains on $(v_{i^*,0}, v^{(3)},v^{(3)})$ is at most $\ell$. 
Recall that by assumption, $(v_{i^*,0}, v^{(3)},v^{(3)})$ and $(v_{i^*,0}, v^{(4)},v^{(4)})$ have the same transcript in the protocol that computes the social welfare.
By the rectangle argument, $(v_{i^*,0}, v^{(3)},v^{(4)})$ must also have the same transcript, and thus the welfare of the allocation from $\mathcal A$ in this case is also at most $\ell$. However, the optimal welfare from $(v_{i^*,0}, v^{(3)},v^{(4)})$ is at least $(\sqrt{3}-1)\ell + 2(\ell - 1)$, obtained by giving $\overline{G}_{i^*}$ to Alice, $G_{i^*} \setminus H_{i^*,j_2^*}^{(1)}$ to Bob, and $H_{i^*,j_2^*}^{(1)}$ to Charlie.

    Thus, the best possible approximation ratio of the mechanism $\mathcal A$ is at most:
    \begin{equation*}\label{eq::to-explain}
      \max\Bigg\{\frac{((\sqrt{3}-1)\ell+1)+2}{2(\ell-1)}, \frac{\ell}{(\sqrt{3}-1)\ell+2(\ell-1)}\Bigg\} < \frac{\sqrt{3}-1}{2} + \frac{3}{\ell}. \tag{for $\ell \ge 3$} 
    \end{equation*}
By that, we get a contradiction, which completes the proof. 
\end{proof}

\subsection{``No bad $4$-tuple'' implies high communication}
\label{subsec::no-bad-implies-high-cc}
We are now ready to finalize the proof of \autoref{thm:truthful-lower-bound-for-subadditive-cup-sm}.  
Note that by \autoref{prop::bad-tuple}, a communication-efficient truthful mechanism with approximation ratio no worse than $\frac{\sqrt 3-1}{2}+\frac{3}{\ell}$ cannot have a ``bad $4$-tuple''. We will now show
$e^{\Omega(\frac{\sqrt{m}}{\log m})}$ bits of communication is necessary for the mechanism $\mathcal{A}$ to avoid such ``bad $4$-tuples''.

    For any fixed transcript $T$ for $\mathcal{P}$, let $C(T)$ be the set of pairs $(\vb{b},\vb{C})\in \{0,1\}^k\times \{0,1\}^{k\times k}$ such that $\mathcal{P}$ uses transcript  $T$ on $(v_{\vb{b},\vb{C}},v_{\vb{b},\vb{C}})$, i.e., $C(T) = \{(\vb{b},\vb{C}) : \tau_{\mathcal{P}}(v_{\vb{b},\vb{C}},v_{\vb{b},\vb{C}}) = T\}$.
    For any index $i \in [k]$ and bit $z \in \{0,1\}$, we further define 
    for each transcript the following subsets of length-$k$ bit-vectors: 
    $$C_{i,z}(T)=\{\vb{C}[i]: (\vb{b},\vb{C}) \in C(T), \vb{b}[i]=z\} \subseteq \{0,1\}^k.$$
    In other words, for every transcript $T$ of the the menu computation protocol $\mathcal{P}$, the subset $C_{i,z}(T)$ contains for every $(\vb b,\vb C)\in C(T)$, the $i$'th rows of $\vb{C}$ only for cases in which $\vb b[i]=z$. Note that $C_{i,z}(T)$ may be empty.

We now claim that for every transcript $T$ of the menu computation protocol, and for every index $i$, it holds that
$|C_{i,0}(T)| \le 1$ or $|C_{i,1}(T)| \le 2^{\cc(\mathcal{SW})}$. The reason for it is
that if the converse holds, i.e. there exists an index $i^\ast$ such that both 
$|C_{i^\ast,0}(T)| > 1$ and $|C_{i^\ast,1}(T)| > 2^{\cc(\mathcal{SW})}$, then there is a bad $4$-tuple.\footnote{To see why, observe that by the pigeonhole principle, if both $|C_{i^\ast,0}(T)|$ and $|C_{i^\ast,1}(T)|$  are greater than $1$, then there exist $(\vb b^{(1)},\vb C^{(1)}),(\vb b^{(2)},\vb C^{(2)}),(\vb b^{(3)},\vb C^{(3)}),(\vb b^{(4)},\vb C^{(4)})$ that satisfy Property \ref{bad-cond1}, \ref{bad-cond2}, \ref{bad-cond3}, and \ref{condi-4} of \autoref{prop::bad-tuple} for this $i^\ast$ and for some $j^\ast_1, j^\ast_2$, since they all share transcript $T$. Furthermore, the fact that $|C_{i^\ast,1}(T)|$ is greater than the number of transcripts for the protocol that computes the social welfare, implies that there are two vectors in $C_{i^\ast,1}(T)$ that also share the same transcript for the protocol $\mathcal{SW}$.}

Note that $|C_{i,0}(T)| \le 1$ or $|C_{i,1}(T)| \le 2^{\cc(\mathcal{SW})}$ implies that:
\begin{equation}\label{eq-ub-c}
 |C_{i,0}(T)| + |C_{i,1}(T)| \le 2^{k} + 2^{\cc(\mathcal{SW})}.
\end{equation}
The explanation for it is by a simple case analysis. If $|C_{i,0}(T)| \le 1$, then: 
$$|C_{i,0}(T)|+|C_{i,1}(T)| \le 1+2^k\le 2^{\cc(\mathcal{SW})}+ 2^{k}.$$ 
If $|C_{i,1}(T)| \le 2^{\cc(\mathcal{SW})}$, then similarly: 
$$|C_{i,0}(T)| + |C_{i,1}(T)| \le 2^{k} + 2^{\cc(\mathcal{SW})}.$$ 

By definition and by plugging in \autoref{eq-ub-c}, we have that:
\begin{equation}\label{eq-ub-bigc}
    |C(T)| \le \prod_{i \in [k]} (|C_{i,0}(T)| + |C_{i,1}(T)|) \le (2^{k} + 2^{\cc(\mathcal{SW})})^k.
\end{equation}
Denote with $\mathcal T$ the subset of transcripts in the menu protocol, i.e. $\mathcal T=\{\tau (v_B,v_C): (v_B,v_C)\in \SubAdd^\ast)\}$.  
Observe that:
\begin{align*}
    2^{k+k^2} &= \sum_{T\in \mathcal T} C(T) \tag{the total number of pairs $(\vb{b},\vb{C})$ is $2^{k+k^2}$} \\ 
    &\le |\mathcal T| \max_{T\in \mathcal T}{C(T)} \\
    &\le 2^{\cc(\mathcal P)} \cdot (2^{k} + 2^{\cc(\mathcal{SW})})^k.
    \tag{by \autoref{eq-ub-bigc}}
\end{align*}
Therefore, we have:
\begin{equation}\label{eq-sw}
2^{\cc(\mathcal{P})} \ge \frac{2^{k+k^2}}{(2^{k} + 2^{\cc(\mathcal{SW})})^k} 
=\left(\frac{2^{1+k}}{2^k + 2^{\cc(\mathcal{SW})}}\right)^k
=\left(\frac{2}{1 + 2^{\cc(\mathcal{SW})-k}}\right)^k.    
\end{equation}
To conclude the proof, it suffices to show that 
$\frac{2}{1 + 2^{\cc(\mathcal{SW})-k}} > B$ for some constant $B > 1$.
We can show this as follows. %
Consider we use an $\ell$-independent and $k$-width family 
with $k=e^{\frac{2\sqrt{m}}{\log m}},\ell=\frac{1}{4} \log m$, which is guaranteed to exist by \autoref{lemma:large-family-existence}.
Let $d=\cc(\mathcal{SW})-\cc(\mathcal A)$. We remind that $d$ is at most polynomial in $m$.
Observe that if $\cc(\mathcal A)\ge k-d-1$, then $\cc(\mathcal A)=\exp(\Omega(\frac{2\sqrt{m}}{\log m}))$, and we are done.
Thus, we can assume that $\cc(\mathcal A)< k-d-1$, which implies that $\cc(\mathcal{SW})\le k-1$. Therefore, a simple computation gives  that:
\begin{equation}\label{eq-sw-bound2}
 \frac{2}{1 + 2^{\cc(\mathcal{SW})-k}} \ge \frac{4}{3}.
\end{equation}
Combining \autoref{eq-sw} and \autoref{eq-sw-bound2} gives that $ \cc(\mathcal{P}) \ge k\log \frac{4}{3}$.  Thus, applying \autoref{lemma:menu-reconstruction-cc} gives that
 $\poly(\cc(\mathcal{A}),m) \ge k \log \frac{4}{3}$, so $\cc(\mathcal{A}) \ge e^{\Omega(\frac{\sqrt{m}}{\log m})}$. This completes the proof that $\cc(\mathcal{A})$ is exponentially high.

\begin{remark}\label{remark:class-of-vals-subadd}
A discerning reader might have observed that in the construction utilized in the proof of \autoref{thm:truthful-lower-bound-for-subadditive-cup-sm}, Alice's valuation set comprises solely single-minded valuations, whereas Bob's and Charlie's valuation sets consist exclusively of subadditive valuations. Consequently, one might be inclined to argue that we demonstrate an impossibility for the scenario featuring one single-minded bidder and two subadditive bidders, i.e., for the valuation class $\SingleM\times \SubAdd\times \SubAdd$. However, this assertion is not accurate, and the lower bound and gap that we show are for the class $(\SingleM\cup \SubAdd)^3$. The rationale behind this lies in the taxation framework which, as we restate in \autoref{lemma:menu-reconstruction-cc}, necessitates all bidders' valuation classes to be ``sufficiently rich'' (in this case, the classes must include subadditive valuations). 
\end{remark}

\section{Supplemental Results: Lower Bounds on Two-Bidder Mechanisms}
\label{section::two-player}

In this section, we give an additional impossibility result for two bidders which holds when bidders might be single-minded.
In fact, the results of this section served as stepping stones towards our approach in Sections~\ref{sec:proof-prelim} and~\ref{sec:real-proof}.

In contrast to the three-bidder hardness in Sections \ref{sec:proof-prelim} and \ref{sec:real-proof} that requires showing hardness of computing the menu which was more involved, here we use the taxation framework in a more direct way. We do not need to show that computing the menu is hard.\footnote{Actually, there are only two players, so computing the menu is a one player problem which cannot possibly be hard.} We prove our impossibility result by showing that the number of menus in every truthful auction has to be high. Interestingly, our bound can also be achieved by arguing about the hardness of simultaneous protocols. We believe that due to the relative simplicity of the construction, this section serves as a way to learn about the structure of mechanisms for bidders both with and without complementarities. We opted to defer this discussion until this section to keep Sections~\ref{sec:proof-prelim} and~\ref{sec:real-proof} self-contained.

We show that for two bidders with valuations in $\XOS \cup \SingleM$, any deterministic truthful mechanism that gives an approximation asymptotically better than $\frac{\sqrt 5-1}{2}$  requires exponential communication.Formally:

\begin{theorem}\label{thm:truthful-lowerbound-for-xos-cup-sm} 
 Every deterministic truthful mechanism $\mathcal{A}$ for two bidders with valuations in $\XOS_m~\cup~\SingleM_m$ that guarantees a $(\frac{\sqrt 5-1}{2}+3m^{-\frac13})$-approximation to the optimal social welfare requires $\exp(\Omega(m^{\frac13}))$ bits of communication. 
\end{theorem}

Note that  the deterministic protocol that $\frac34$-approximates the optimal welfare for $\XOS$ \cite{Fei09, DNS05-journal} can be generalized to work for $\XOS \cup \SingleM$ by using \autoref{lemma::black-box-upper-bound}. Therefore, \autoref{thm:taxation-lower-bound-xos} indeed implies a separation.

The proof of \autoref{thm:taxation-lower-bound-xos} is based on a direct taxation argument, i.e. we show that the number of menus of every truthful mechanism that gives a \textquote{good} approximation has to be doubly exponential in $m$. 
To the best of our knowledge, this is not easy to do with the lower bound construction of \cite{AKSW20}.\footnote{\autoref{thm:taxation-lower-bound-xos} can also be proven by showing
a lower bound on simultaneous protocols, which is the same argument used in \cite{AKSW20}. See \autoref{rem:alternative-proof-for-xos-cup-sm}.} We begin by describing the class of valuations that we use to show the separation (\autoref{sec:xos-lb-truth-class-of-valuation}). We then prove \autoref{thm:taxation-lower-bound-xos} in \autoref{sec:proof-of-xos-lb-truth}. 
In \autoref{subsec::difficulties-xos}, we explore the challenges associated with directly generalizing \autoref{thm:taxation-lower-bound-xos} to a scenario involving three bidders. By that, we illustrate the challenges encountered when proving \autoref{thm:truthful-lower-bound-for-subadditive-cup-sm}.

En route, we explore the structure of payments of approximately optimal and truthful mechanisms, showing 
that as the approximation guarantee of a mechanism goes to $1$, then its payments approach VCG payments. We state this formally in \autoref{app:structure-of-payments}.

\subsection{Description of the Class of Valuations}\label{sec:xos-lb-truth-class-of-valuation}

Recall the definition of XOS valuations:

\begin{definition}[XOS valuations]
    A valuation function $v:2^{[m]}\to \mathbb{R}$ is \emph{XOS} if there exists a collection of additive clauses $\mathcal{C} \subseteq \mathbb{R}^m$ such that for all $S \in 2^{[m]}$, $v(S)=\max_{\vb{c} \in \mathcal{C}} \sum_{i \in S} c_i$. We denote the family of all XOS valuation functions over $m$ items by $\XOS=\XOS_m$.
\end{definition} 
An $XOS$ valuation is \emph{binary} if all its clauses are 0/1 valued, i.e., for every $c\in \mathcal C$, $c_i\in \{0,1\}$. 
Equivalently, binary $XOS$ valuations can be defined by a collection of subsets of $[m]$: 
\begin{definition}[Binary XOS valuations]
    A valuation function $v:2^{[m]}\to \mathbb{R}$ is {\em binary XOS} if there exists a collection $\mathcal{G} \subseteq 2^{[m]}$ such that $v(S)=\max_{G \in \mathcal{G}}\{|S \cap G|\}$ for all $S \in 2^{[m]}$. Denote the family of all binary XOS valuation functions over $m$ items by $\mathsf{BXOS}=\mathsf{BXOS}_m \subset \XOS_m$.
\end{definition}

To prove \autoref{thm:truthful-lowerbound-for-xos-cup-sm}, we will consider a special subset of XOS valuations, which takes sets from a collection that satisfies the \emph{average intersection property}:

\begin{definition}
A collection $\mathcal G \subseteq 2^{[m]}$ satisfies the \emph{$b$-average intersection property} for some $b \in (0,1)$ if every $G\in\mathcal G$ is of size $bm$, and for every $G_1\neq G_2\in \mathcal G$, $|G_1 \cap G_2|\le 2b^2m$. 
\end{definition}

By the probabilistic method, such collections always exist and can be made exponentially large. Formally:
\begin{lemma}\label{lem:xos-avg-intersection}
    For any $b \in (0,1)$,
    there exists a collection $\mathcal G \subseteq 2^{[m]}$ of size $\exp(b^2m/6)$ that satisfies the $b$-average intersection property.
\end{lemma}
The proof of \autoref{lem:xos-avg-intersection} can be found in \autoref{app:missing-proof-sec-5}.
Let $b=m^{-\frac{1}{3}}$ and $\mathcal G$ be such a collection containing $\exp(m^{\frac{1}{3}}/6)$ subsets of $[m]$ that satisfies $m^{-\frac13}$-average intersection property.
We define a special set of binary XOS valuations $\XOS^*=\{v_\mathcal{H} : \mathcal{H} \subseteq \mathcal{G}\} \subseteq \mathsf{BXOS}_m$ where
$$v_\mathcal{H}(S)=\max_{H \in \mathcal{H}} |S \cap H|~~~~\text{for all}~S \in 2^{[m]}.$$

Meanwhile, we consider the class of single-minded valuations $\SingleM^*\subseteq \SingleM_m$ containing single-minded valuations with constant value $\alpha$, i.e., $\SingleM^* = \{v_{T,\delta}: T \subseteq [m], \delta \in \{0,1\}\}$ where
$$v_{T,\delta}(S)=(\alpha+\delta) \cdot \mathrm{I}[S \supseteq T]~~~~\text{for all}~S \in 2^{[m]}.$$
and $\alpha$ is a constant that we fix later.

Later, the valuation profiles $\{(v_A,v_B): v_A \in \SingleM^*, v_B \in \XOS^*\}$ will be used to show a lower bound on the communication complexity.

\subsection{Proof of \autoref{thm:truthful-lowerbound-for-xos-cup-sm} via Taxation Complexity}
\label{sec:proof-of-xos-lb-truth}

In this section, we will present a proof of \autoref{thm:truthful-lowerbound-for-xos-cup-sm} based on the Taxation Complexity framework~\cite{D16b}.

\begin{definition}[Taxation complexity]\label{def::tax-comp}
    Let $\mathcal{A}=(f,p_1,\ldots,p_n):\V_1\times \cdots\times\V_n \to \Alloc \times \R^n$ be any deterministic truthful mechanism.
    The \emph{taxation complexity} $\tax(\mathcal{A})$ of $\mathcal{A}$ is defined as the number of bits needed to represent the index of a specific menu, i.e., $\tax(\mathcal{A})=\max_{i \in [n]} \log |\{\Menu^{\mathcal{A}}_i(v_{-i}) : v_{-i} \in \V_{-i}\}|$.
\end{definition}

\begin{theorem}[\cite{D16b}]\label{thm:taxation-lower-bound-xos}
    Let $\mathcal{A}=(f,p_1,\ldots,p_n):\V_1\times \cdots\times\V_n \to \Alloc \times \R^n$ be any deterministic truthful mechanism, where all domains $\V_1,\ldots,\V_n$ contain $\XOS_m$. Then, $\cc(\mathcal A) \ge \frac{\tax(\mathcal A)}{m}-1$. 
\end{theorem}

Observe that \autoref{thm:taxation-lower-bound-xos} only requires that the domain of valuations contains $\XOS$, and it remains applicable for mechanisms that are truthful for a larger domain, e.g., $\XOS \cup \SingleM$.
Therefore, our first step in the proof \autoref{thm:truthful-lowerbound-for-xos-cup-sm} will be to show a lower bound on the taxation complexity of every mechanism that gives a \textquote{good} approximation.
Formally:
\begin{proposition}\label{prop:taxation-lowerbound-for-xos-cup-sm}
Let $\V = \XOS_m\cup\SingleM_m$ and $\mathcal{A}=(f,p_A,p_B): \V \times \V \to \Alloc \times \R^2$ be any deterministic truthful mechanism. 
If $\mathcal{A}$ obtains a $(\frac{\sqrt 5-1}{2}+3m^{-\frac13})$-approximation to the optimal welfare, then $\tax(\mathcal{A})\ge \exp(m^{\frac13}/6)$.
\end{proposition}

\begin{proof}[Proof of \autoref{prop:taxation-lowerbound-for-xos-cup-sm}]
Suppose that the mechanism $\mathcal{A}$ indeed obtains an approximation to the welfare of at least $(\frac{\sqrt 5-1}{2}+3m^{-\frac13})$.
Consider the case when Alice's valuation is $v_A \in \SingleM^*$ and Bob's valuation $v_B \in \XOS^*$.
We want to show that for every distinct pair of valuations of Bob $v_B, v_B' \in \XOS^*$, the menus that are presented to Alice must be different, i.e., $\Menu_A^\mathcal{A}(v_B) \ne \Menu_A^\mathcal{A}(v_B')$. Then, $|\{\Menu_A^\mathcal{A}(v_B): v_B \in \XOS^*\}|=|\XOS^*|=2^{|\mathcal{G}|}$ and the proposition follows.
For simplicity, from now on we denote $\Menu_A^\mathcal{A}(\cdot)$ with $\Menu(\cdot)$.

For any distinct $v_B,v_B' \in \XOS^*$, denote with $\mathcal H$ and with $\mathcal H'$ respectively the combinations that are associated with $v_B$ and with $v_B'$ respectively. 
Since $v_B\ne v_B'$, there exists at least one subset $H$ such that (without loss of generality) $H\in \mathcal H$, whereas $H\notin \mathcal H'$.
Suppose $\Menu(v_B)=\Menu(v_B')=M$ by contradiction.
There are two possibilities:

\begin{itemize}
    \item If $M(\overline{H}) \le \alpha$, then given the valuation profile $(v_A=v_{\overline{H},1},v_B')$, Alice always gets $\overline{H}$. The actual welfare the mechanism $\mathcal A$ gets is at most $v_A(\overline{H})+v_B'(H) \le (\alpha +1) + 2b^2 m$, while the optimal welfare being at least $v_B'([m])=bm$.
    \item If $M(\overline{H}) > \alpha$, then given the valuation profile $(v_A=v_{\overline{H},0},v_B)$, Alice does not get $\overline{H}$. The actual welfare the mechanism $\mathcal A$ gets is at most $v_B([m])=bm$, while the optimal welfare being $v_A(\overline{H})+v_B(H)=\alpha + bm$.
\end{itemize}
To summarize, whenever $\Menu(v_B)=\Menu(v_B')$ for some distinct $v_B,v_B' \in \XOS^*$, we know that the approximation ratio of the mechanism to the optimal welfare can be at most
\begin{equation}\label{eq-approx-ratio-tax-xos}
  \max\left\{\frac{(\alpha+1)+2b^2m}{bm},\frac{bm}{\alpha+bm} \right\}.  
\end{equation}
Recall $b=m^{-\frac13}$ and let $\alpha= \frac{\sqrt{5}-1}{2}  bm = \frac{\sqrt{5}-1}{2} m^{\frac23}$.
Then, the ratio above is always strictly smaller than $\frac{\sqrt{5}-1}{2} + 3m^{-\frac13}$.
Therefore, to get a $(\frac{\sqrt{5}-1}{2} + 3m^{-\frac13})$-approximation, $\mathcal{A}$ must have different menus for all $v_B \in \XOS^*$. As a result, $\tax(\mathcal{A})\ge \log |\XOS^*| = |\mathcal{G}|=\exp(m^{\frac13}/6)$.
\end{proof}

Having \autoref{prop:taxation-lowerbound-for-xos-cup-sm}, \autoref{thm:truthful-lowerbound-for-xos-cup-sm} follows directly via the Taxation Complexity framework.

\begin{proof}[Proof of \autoref{thm:truthful-lowerbound-for-xos-cup-sm}]
    Since $\XOS_m \cup \SingleM_m \supseteq \XOS_m$, we have $\cc(\mathcal{A}) \ge \frac{\tax(\mathcal{A})}{m}-1$ by \autoref{thm:taxation-lower-bound-xos}.
    Further by \autoref{prop:taxation-lowerbound-for-xos-cup-sm}, we have $\cc(\mathcal{A}) \ge \exp(\Omega(m^{\frac13}))$.
\end{proof}

\begin{remark}\label{rem:alternative-proof-for-xos-cup-sm}
    \autoref{thm:truthful-lowerbound-for-xos-cup-sm} can be also proven via a communication lower bound for simultaneous protocols combining with the two-bidder corollary from Taxation Complexity framework~\cite{D16b}, as used in~\cite{AKSW20}.
    The proof follows a similar structure, and we include it in \autoref{app:simultaneous-proof-for-xos-cup-sm} for completeness.
\end{remark}
\begin{remark}
    Similarly to \autoref{remark:class-of-vals-subadd}, one might think that the construction implies that our impossibility holds for a single-minded bidders and a bidder with XOS valuations, i.e. for the class $\SingleM\times \XOS$. However,  our use of \autoref{thm:taxation-lower-bound-xos} implies that our the lower bound applies solely to the valuation class $(\SingleM\cup \XOS)^2$. In particular, the VCG mechanism can be implemented with $\poly(m)$ communication for $\SingleM\times \XOS$.\footnote{This is because the optimal allocation can be computed with $\poly(m)$ communication: Alice simply sends her desired set $S$ together with its value, and Bob replies whether the optimal allocation is $(S,\overline S)$ or $(\emptyset,M)$.}
\end{remark}

\subsection{Difficulties for a Stronger Three-Bidder Lower Bound} \label{subsec::difficulties-xos}

Naturally, one might wonder whether the class of valuations defined in \autoref{sec:xos-lb-truth-class-of-valuation} can be generalized to a stronger (i.e., better than $\frac{\sqrt{5}-1}{2}$) lower bound beyond two bidders, since the Taxation Complexity framework extends beyond two bidders. 
We now describe an attempt to generalize our construction and briefly discuss why it is not likely to work.

Recall that in \autoref{sec:proof-of-xos-lb-truth}, the subset of valuations $\{(v_A,v_B): v_A \in \SingleM^*, v_B \in \XOS^*\}$ was used as hard instances for the taxation complexity lower bound, where $\XOS^*$ is defined by some collection $\mathcal{G}$ that satisfies the $b$-average intersection property.
To get a stronger three-bidder lower bound,
we add another bidder named Charlie, and
consider the valuations $\{(v_A,v_B,v_C): v_A \in \SingleM^*, v_B,v_C \in \XOS^*\}$. 
Our hope is to get a stronger lower bound by \emph{the extra communication hardness from welfare maximization between these two XOS bidders}. 

Our plan is to use a similar argument to the one used in the proof of \autoref{prop:taxation-lowerbound-for-xos-cup-sm}.
For convenience, let us assume $\mathcal{G}$ consists of exponentially many random sets of size $bm/2$ for some constant $b\in (0,1)$. Similar to the proof of \autoref{lem:xos-avg-intersection}, the collection $\mathcal{G}$ satisfies $(b/2)$-average intersection property with high probability and its size is $\exp(m)$. 
Consider four valuations $v^{(1)},v^{(2)},v^{(3)},v^{(4)}$ sampled uniformly at random from $\XOS^*$ and associated with the sub-collections $\mathcal{H}^{(1)},\mathcal{H}^{(2)},\mathcal{H}^{(3)},\allowbreak\mathcal{H}^{(4)} \subseteq \mathcal{G}$, respectively.

We argue without proof that with high probability, there exist subsets $H_1, H_2,H_3,H_4$ such that for every $i\in [4]$:
\begin{enumerate}
    \item $H_i \in \mathcal{H}^{(i)}$.
    \item For every $j\neq i$, $H_i \notin  \mathcal{H}^{(j)}$. 
    \item $H_1\cap H_2=H_3\cap H_4=\emptyset$.
\end{enumerate}
As a result:
\begin{itemize}
\item Given the valuation profile $(v_B=v^{(1)},v_C=v^{(2)})$, the maximum welfare $bm$ is obtained by allocating $H_1 \cup H_2$ to them. 
\item Given the valuation profile  $(v_B=v^{(3)},v_C=v^{(4)})$, the welfare obtained from allocating $H_1 \cup H_2$ to Bob and Charlie
is at most $2b^2m$, 
due to the $(b/2)$-average intersection property. Maximum welfare of $bm$ can be obtained from allocating $H_3 \cup H_4$ to them.
\end{itemize}

Let $v_A \in \SingleM^*$ be the single-minded valuation that has value $\alpha$ for the set  $[m]\setminus (H_1 \cup H_2)$. 
Assume towards a contradiction that $\Menu(v^{(1)},v^{(2)})=\Menu(v^{(3)},v^{(4)})$ and denote with $P$ the price in both menus for the bundle $[m]\setminus (H_1\cup H_2)$. Assume that $P \ne \alpha$ (which is without loss of generality, since we can always slightly perturb the value of $\alpha$). 
Note that it must be the case that either:
\begin{itemize}
    \item If $P<\alpha$, then Alice with value $v_A$ gets the bundle $[m]\setminus (H_1\cup H_2)$: then for $(v_A,v^{(3)},v^{(4)})$, the mechanism gets welfare of at most $\alpha+2b^2m$, while the optimal welfare is $v_B^{(3)}(H_3)+v^{(4)}(H_4)=bm$.
    \item If $P>\alpha$, then
    Alice with valuation $v_A$ does not get the bundle $[m]\setminus (H_1\cup H_2)$: then for $(v_A,v^{(1)},v^{(2)})$, the mechanism gets at most $(b-b^2/4)m$ \emph{due to the communication hardness of allocating $[m]$ between Bob and Charlie}\footnote{The welfare of $(b-b^2/4)m$ is obtained 
    by giving any set $S \in \mathcal{H}^{(1)}$ to Bob and the rest of the items to Charlie. The $(b/2)$-average intersection property implies that this allocation has welfare of at least $(b-b^2/4)m$.
    The approximation ratio $\frac{(b-b^2/4)m}{bm}$ is smallest when $b=1$, which is the most difficult case for welfare maximization between Bob and Charlie in general.}, while the optimal welfare is $\alpha+bm$.
\end{itemize}
Overall, the approximation ratio to the optimal welfare can be at most 
$$\max\left\{\frac{\alpha+2b^2m}{bm}, \frac{(b-b^2/4)m}{\alpha+bm}\right\},$$
which is quite very similar to \autoref{eq-approx-ratio-tax-xos} in the proof of \autoref{thm:taxation-lower-bound-xos}.
The only difference is the numerator of the second term gets smaller (from $bm$ to $(b-b^2/4)m$) which is a significant difference when $b$ is close to $1$).
However, the minimum is still obtained when $b\to 0$ and $\alpha=\frac{\sqrt{5}-1}{2}bm$, resulting in no improvement at all.

As a result, the additional difficulty of allocating items among two bidders instead of one does not help, at least for this specific class of valuations. 
The reason for it is that  
for our original argument that gets $\frac{\sqrt{5}-1}{2}$ for two bidders, we want $b \to 0$ to ensure the intersection of a random clauses of the XOS bidder with the complement of the single-minded bidder to be as small as possible. However, 
to get a significant gap in communication hardness for two XOS bidders, we want $b\to 1$ since the intersection of a random clauses of them needs to be large, so that one cannot get a good welfare between them easily.
Nevertheless, this idea turned out to be useful for the proof of  \autoref{thm:truthful-lower-bound-for-subadditive-cup-sm}, 
by using subadditive modified set-cover valuations.

\section*{Acknowledgements}
The authors thank Shahar Dobzinski for the helpful discussions throughout the duration of this work.

\bibliographystyle{alpha}
\bibliography{sample}

\appendix

\section{Missing Proofs}
\label{app:missing-proofs}
\subsection{Missing Proofs of \autoref{sec:prelims}}
\subsubsection{Proof of \autoref{lemma::black-box-upper-bound}}
Let $\mathcal P$ be a protocol that achieves an $\alpha$-approximation to the optimal welfare with bidders whose valuations are in class $\mathcal{C}$.
We first describe the protocol $\mathcal P'$ for $\mathcal{C}\cup \SingleM$, and then prove that its approximation guarantee is also $\alpha$. 

\paragraph{Description of Protocol.}
  Briefly, the protocol operates by simply running $\mathcal{P}$ independently once for each possible method of allocating the single-minded bidders their preferred bundle.

  In more detail, first, each bidder $i$ sends a bit that specifies whether $v_i$ is single-minded or belongs in $\mathcal C$ (if $v_i \in \mathcal C \cap \SingleM$, it is considered single-minded). Let $\mathsf{SM}\subseteq [n]$ be the set of bidders such that $v_i$ is single-minded, and with $\mathsf{Other}$ the rest of the bidders.   Then, we ask all the bidders in $\mathsf{SM}$ to send their (minimal) desired set and their value for it. We say that a subset of single-minded bidders $K\subseteq \mathsf{SM}$ is \emph{feasible} if for every $i,j\in K$, $S_i\cap S_j=\emptyset$. Note that at this point in the protocol, we can identify all the feasible subsets of single-minded bidders. 

  Next, for every feasible subset of single-minded bidders $K \subseteq \mathsf{SM}$, we consider the following allocation. First, we allocate to the bidders in $K$ their desired sets. 
The rest of the single-minded bidders are allocated with the empty bundle.
  Then, we allocate the remaining items to the bidders in $\mathsf{Other}$, using the protocol $\mathcal P$. Afterwards, all bidders send their value for the bundle given this allocation, so its social welfare is known.  We conclude by outputting the allocation with the maximum welfare. 
It is easy to see that the protocol $\mathcal{P}'$ requires at most $\poly(\cc(\mathcal{P}),m,2^n)$ bits.

\paragraph{The Approximation Ratio}
Informally, the approximation ratio $\alpha$ still holds for $\mathcal P'$ because we \emph{exactly} optimize welfare for all single-minded bidders, and on the remaining bidders, we can still only be off by a factor of at most $\alpha$. 

In more detail, denote with $A=(A_1,\ldots,A_n)$ the allocation that the protocol outputs and with $O=(O_1,\ldots,O_n)$ an optimal allocation. 
Assume without loss of generality that  $(O_1,\ldots,O_n)$ and $(A_1,\ldots,A_n)$ satisfy that no bidder is allocated any item that does not strictly increase her value.
Denote with $K^\ast$ the subset of single-minded bidders that get a valuable bundle in the optimal allocation, i.e. $K^\ast=\{i : v_i(O_i)>0, i\in \mathsf{SM}\}$.  
Denote with $X=(X_1,\ldots,X_n)$ the allocation at the iteration where the feasible set of single-minded bidders is $K$. Note that by definition for every single-minded bidder $i$, $X_i=O_i$. Therefore,
\begin{equation}\label{eq-xo-same}
\sum_{i\in \mathsf{SM}} v_{i}(X_i)= \sum_{i\in \mathsf{SM}} v_{i}(O_i)  
\end{equation}
Also, note that 
$\bigcup_{i\in \mathsf{SM}} X_i=
\bigcup_{i\in \mathsf{SM}}O_i$, and we denote this subset of items with $J_{\mathsf{SM}}$.  Observe that since $O$ is a welfare-maximizing allocation, it is in particular an optimal allocation of the items in $[m]\setminus J_{\mathsf{SM}}$ to the bidders in $\mathsf{Other}$. We remind that the
allocation of the bidders in $\mathsf{Other}$ is determined by the protocol $\mathcal P$. 
Due to the approximation guarantee, we get that: 
\begin{equation}\label{eq-xos-approx}
    \sum_{i\in \mathsf{Other}}v_i(X_i)\ge \alpha \cdot \big(\sum_{i\in \mathsf{Other}} v_i(O_i) \big) 
\end{equation}
Therefore, we can deduce that:
\begin{align*}
    \sum_{i\in [n]} v_i(A_i)&\ge\sum_{i\in[n]} v_i(X_i)  &\text{(by construction)} \\
    &= \sum_{i\in \mathsf{SM}}v_i(X_i)+ \sum_{i\in \mathsf{Other}}v_i(X_i) \\
    &\ge \sum_{i\in \mathsf{SM}}v_i(O_i)
    + \alpha\cdot \sum_{i\in \mathsf{Other}}v_i(O_i)  &\text{(by (\ref{eq-xo-same}) and (\ref{eq-xos-approx}))} \\
    &\ge  \alpha\cdot\sum_{i \in [n]}v_i(O_i)
\end{align*}
which completes the proof. 

\subsubsection{Proof of \autoref{lemma:menu-reconstruction-cc}}
We begin by redefining $\price(\mathcal A)$ and $\tax(\mathcal A)$. The definitions are identical to the ones in \cite{D16b}, and we write them for the sake of completeness. 
$\price(\mathcal A)$ is the communication complexity of the $(n-1)$-bidder problem of computing the price for a specific bundle $S\subseteq M$ in the menu presented to player $i$ in the mechanism $\mathcal A$ given the valuation profile $v_{-i}$. For the definition of $\tax(\mathcal A)$, see \autoref{def::tax-comp}. 

Now, recall that  by the menu reconstruction theorem  \cite[Theorem 3.1]{D16b}, the communication complexity of computing the menu presented to a player $i$ given the valuations $v_{-i}$ of the other players and the truthful mechanism $\mathcal A$ is $\poly(\tax(\mathcal A),\price(\mathcal A),m,n)$. This is true for every domain.
In addition, by \cite[Proposition F.1]{D16b}, it holds that  $\price(\mathcal A)\le \cc(\mathcal A)$ for every truthful mechanism for a class of valuations that includes additive valuations. Similarly,  \cite[Proposition 2.3]{D16b} implies that 
$\tax(\mathcal A)\le \cc(\mathcal A)$ for every truthful mechanism $\mathcal A$ that is truthful for a class of valuations that contains subadditive valuations.\footnote{In fact, the statement in \cite{D16b} is specifically for the class of subadditive valuations. However, the proof holds as-written for every class of valuations that contains subadditive valuations.}
The proof is obtained by combining these assertions.

\subsection{Missing Proofs of \autoref{sec:real-proof}} \label{app:missing-proof-sec-4}
\begin{proof}[Proof of \autoref{lemma:large-family-existence}]
    We will show that for sufficiently large $m$, there exists an 
    $\ell$-independent $k$-width-family, where $\ell=\frac{1}{4} \log m$ and $k=e^{\frac{2\sqrt{m}}{\log m}}$.

        To obtain a $\ell$-independent $k$-width-family $\mathcal{F}=(\mathcal{G},\{\mathcal{H}_i^{(t)}\}_{t \in \{0,1\}, i \in [k]})$, consider the following two-layered randomized construction:
    \begin{enumerate}
        \item \emph{First Layer}: For every $i \in [k]$, let $G_i$ be the random set that contains each item $e \in [m]$ independently with probability ~$\frac12$.
        \item \emph{Second Layer}: For every $i, j \in [k]$, let $H_{i,j}^{(0)}$ be the random set that contains each $e \in \overline{G_i}$ independently with probability ~$\frac12$. Similarly, let $H_{i,j}^{(1)}$ be the random set that contains each item $e \in G_i$ independently  with probability ~$\frac12$. 
        Note that for every $i$ and for every $t\in\{0,1\}$, the distribution of the items in $H_{i,j}^{(t)}$ is identical for all values of $j$.
    \end{enumerate}
Thus, by construction $\mathcal F$ is a $k$-width-family.      
Thus, it remains to prove that with a non-zero probability, $\mathcal{F}$ is $\ell$-independent.

    Recall that $\mathcal{F}[\vb{b},\vb{C}]=\{S_{i,j}\}_{i,j\in [k]}$ is a collection of $k^2$ subsets.
    We can fix an index set $I \subset [k] \times [k]$ of size $\ell$, which specifies a sub-collection $\mathcal{F}[\vb{b},\vb{C}; I]=\{S_{i,j}\}_{(i,j) \in I}$ of size $\ell$ accordingly.
    Our plan is to show that for every choice of $\vb{b},\vb{C}$ and $I$, it is very unlikely that the sub-collection $\mathcal{F}[\vb{b},\vb{C}; I]$ covers $[m]$, and then the existence of a valid $\mathcal{F}$ follows from union bound.
    
    Observe that for every item $e$ and $S_{i,j}$, we have that $e\in S_{i,j}$ with probability $3/4$. In particular, it holds for all the subsets $S_{i,j}$ in $\mathcal{F}[\vb{b},\vb{C}; I]$. Also, note that the events $\{e \in \cup \mathcal{F}[\vb{b},\vb{C}; I]\}_{e \in [m]}$ are mutually independent.  
    Therefore, we have that:
    $$\Pr\left[\bigcup \mathcal{F}[\vb{b},\vb{C}; I] = [m]\right] = \prod_{e \in [m]} \left(1-\Pr\left[e \notin \bigcup \mathcal{F}[\vb{b},\vb{C}; I]\right]\right) = (1-\frac{1}{4^\ell})^m \le (1-2^{-2\ell})^m.$$
    Finally by union bound: 
    \begin{align*}
      \Pr\left[\exists \vb{b},\vb{C},I \text{ such that } \bigcup \mathcal{F}[\vb{b},\vb{C}; I] = [m]\right] 
      &\le\sum_{I \subset [k]\times [k], |I|=\ell} \Pr\left[\bigcup \mathcal{F}[\vb{b},\vb{C}; I] = [m]\right]\\
      &\le \binom{k^2}{\ell}2^{2\ell}(1-2^{-2\ell})^m\\
      &<  k^{2\ell} \exp(-2^{-2\ell}m) = \exp(2\ell \ln k - 2^{-2\ell} m) = 1.
    \end{align*}
    Therefore, with non-zero probability, $\mathcal{F}[\vb{b},\vb{C}]$ is $\ell$-sparse for all $\vb{b},\vb{C}$.
\end{proof}

\subsection{Missing Proofs of \autoref{section::two-player}}\label{app:missing-proof-sec-5}

\begin{proof}[Proof of \autoref{lem:xos-avg-intersection}]
The proof is based on the probabilistic method. We construct the collection $\mathcal G$ by sampling uniformly at random $\exp(b^2m/6)$ subsets of $[m]$, where the size of each subset is $bm$. 
For any two subsets $G_i,G_j$ of size $bm$ sampled uniformly at random, note that the expected size of their intersection $\E{|G_i \cap G_j|} = b^2m$.
By Chernoff bound, we have $\Pr[ |G_i \cap G_j| > 2b^2m] \le \exp(-b^2m/3)$. By applying the union bound, we conclude that: 
\begin{equation*}
\Pr[\exists i \ne j, |G_i\cap G_j| > 2b^2m] \le \sum_{i \ne j} \Pr[ |G_i \cap G_j| > 2b^2m] < |\mathcal{G}|^2 \cdot \exp(-b^2m/3) \le 1. \qedhere  
\end{equation*}
\end{proof}

\subsection{Proof of \autoref{thm:truthful-lowerbound-for-xos-cup-sm} by Hardness of Simultaneous Protocols}\label{app:simultaneous-proof-for-xos-cup-sm}

Recall the following two-bidder corollary from the Taxation Complexity framework:

\begin{theorem}[\cite{D16b}]\label{thm:truthful-implies-simultanous-for-xos}
    Let $\mathcal{A}=(f,p_1,p_2):\V_1\times \V_2 \to \Alloc \times \R^n$ be any deterministic truthful mechanism that guarantees an $\alpha$-approximation to the optimal welfare, and both domains $\V_1,\V_2$ contain $\XOS_m$. Then, there exists a simultaneous protocol $f':\V_1\times\V_2\to\Alloc$ that also guarantees an $\alpha$-approximation with simultaneous communication complexity $\cc(f')\le \poly(\cc(\mathcal{A}),m)$.
\end{theorem}

Based on this result, it suffices to give a lower bound of the communication complexity for simultaneous algorithms to prove \autoref{thm:truthful-lowerbound-for-xos-cup-sm}.

\begin{proposition}\label{prop:simultaneous-lowerbound-for-xos-cup-sm}
Let $f: \SingleM \times \XOS \to \Alloc$ be any deterministic simultaneous protocol. 
If $f$ obtains a $(\frac{\sqrt 5-1}{2}+3m^{-\frac13})$-approximation to the optimal welfare, then its simultaneous communication complexity $\cc(f)\ge \exp(m^{\frac13}/6)$.
\end{proposition}

\begin{proof}
Let Alice be the single-minded bidder and Bob be the XOS bidder.
Since $f$ is a simultaneous algorithm, if Bob sends an identical message upon different valuations $v_B' \ne v_B$, then the resulting allocation will also be identical, i.e.,  $f(v_A,v_B)=f(v_A,v_B')$ for all $v_A \in \SingleM$.
We will show that it further leads to a contradiction with the approximation ratio.

For any distinct $v_B,v_B' \in \XOS^*$, denote with $\mathcal H$ and with $\mathcal H'$ respectively the combinations that are associated with $v_B$ and with $v_B'$ respectively. 
Since $v_B\ne v_B'$, there exists at least one subset $H$ such that (without loss of generality) $H\in \mathcal H$, whereas $H\notin \mathcal H'$.
Consider the case when Alice (with valuation $v_A$) gives a value of $\alpha<bm$ to the bundle $\overline{H}=[m] \setminus H$ (and zero otherwise).
If by contradiction, $f(v_A,v_B)=f(v_A,v_B')$, then Alice must be given a same bundle $A\subseteq [m]$ for both instances $(v_A,v_B)$ and $(v_A,v_B')$. There are two possibilities:
\begin{itemize}
    \item If $\overline{H} \subseteq A$ (i.e., Alice gets $\overline{H}$): for the instance $(v_A,v_B')$, the actual welfare $f$ gets is at most $v_A(\overline{H})+v_B'(H) \le \alpha + 2b^2 m$, while the optimal welfare being at least $v_B'([m])=bm$.
    \item If $\overline{H} \not\subseteq A$ (i.e., Alice gets nothing valuable): for the instance $(v_A,v_B)$, the actual welfare $f$ gets is at most $v_B([m])=bm$, while the optimal welfare being $v_A(\overline{H})+v_B(H)=\alpha + bm$.
\end{itemize}

Overall, we know that the approximation ratio of $f$ can be at most
$$\max\left\{\frac{\alpha+2b^2m}{bm},\frac{bm}{\alpha+bm} \right\}.$$

Recall $b=m^{-\frac13}$ and let $\alpha= \frac{\sqrt{5}-1}{2}  bm = \frac{\sqrt{5}-1}{2} m^{\frac23}$.
Then, the ratio above is always strictly smaller than $\frac{\sqrt{5}-1}{2} + 3m^{-\frac13}$.
Therefore, to get a $(\frac{\sqrt{5}-1}{2} + 3m^{-\frac13})$-approximation, Bob must send different messages for all $v_B \in \XOS^*$. As a result, $\cc(f)\ge \log |\XOS^*| = |\mathcal{C}|=\exp(m^{\frac13}/6)$.
\end{proof}

\begin{proof}[Proof of \autoref{thm:truthful-lowerbound-for-xos-cup-sm}]
    By \autoref{thm:truthful-implies-simultanous-for-xos}, there is some simultaneous protocol $f':\V \times \V \to \Alloc$ that $(\frac{\sqrt 5-1}{2}+3m^{-\frac13})$-approximates the optimal welfare with simultaneous communication complexity $\cc(f') \le \poly(\cc(\mathcal{A}),m)$.
    However, by \autoref{prop:simultaneous-lowerbound-for-xos-cup-sm}, we know $\cc(f') \ge \exp(m^{\frac13}/6)$. Therefore, we conclude that $\cc(\mathcal{A}) \ge \exp(\Omega(m^{\frac13}))$.
\end{proof}

\begin{remark}
We now show that simultaneous protocols for two bidders can achieve $\frac{\sqrt{5}-1}{2}$ of the welfare for the class $\SingleM \times \XOS$, meaning that 
\autoref{prop:simultaneous-lowerbound-for-xos-cup-sm} is asymptotically tight. 
Denote the bundle that Alice (the single-minded bidder) wants by $S$.
Consider some $f^*: \SingleM \times \XOS \to \Alloc$ that gives the bundle $S$ to Alice and the bundle $[m]\setminus S$ to Bob when $v_1(S)\ge \frac{\sqrt{5}-1}{2} v_2([m])$, and otherwise allocates all items  to Bob. 
$f^*$ can be implemented by simultaneous communication and always gives a $\frac{\sqrt{5}-1}{2}$-approximation to welfare.
\end{remark}

\section{The Structure of Payments of Truthful Mechanisms}
\label{app:structure-of-payments}

In this section, we show a formula for
an upper bound and lower bound on 
the payments of approximately optimal and truthful mechanisms whose class of valuations includes single-minded bidders.  
An asymptotic interpretation of \autoref{prop-tax-approx}, which might be of independent interest is that as the approximation guarantee $\alpha$ of a truthful mechanism  goes to $1$, then the payments of the mechanism are approach VCG payments. Moreover, 
\autoref{prop-tax-approx} in its current form applies to two-bidder mechanisms, but it can readily be extended to settings with an arbitrary number of bidders, as well as to the related setting of multi-unit auctions.

We start by defining a class of $\SingleM$ valuations which allow us to bound the prices that any two-player good-approximation mechanism must charge:

\begin{definition}
\label{def:upper-lower-vals}
Fix a valuation $v$, a nonempty subset of items $S\subseteq [m]$ and two constants $\alpha\in [0,1]$ and $\epsilon>0$. We
define a valuation $\upr_{v,\alpha,\epsilon,S}$ as follows:
$$
\forall X\subseteq [m], \quad
\upr_{v,\alpha,\epsilon,S}(X)=
\begin{cases}
    \frac{1}{\alpha}\cdot v([m])-v([m]\setminus S)+\epsilon \quad &X\supseteq S, \\
0\quad &\text{otherwise.} 
\end{cases}
$$
We define $\lwr_{v,\alpha,\epsilon,S}$ similarly:  
$$
\forall X\subseteq [m], \quad
\lwr_{v,\alpha,\epsilon,S}(X) =\begin{cases}
    \max\{{\alpha}\cdot v([m])-v([m]\setminus S)-\epsilon,0\}
    \quad &X\supseteq S, \\
0\quad &\text{otherwise.} 
\end{cases}
$$
\end{definition}

Now, we give the bound on prices enabled by \autoref{def:upper-lower-vals}:

\begin{proposition}\label{prop-tax-approx}
Let $\mathcal A=(f,p_1,p_2):\V_1\times \V_2\to \Alloc \times \mathbb{R}^2$ be any deterministic truthful mechanism that gives an $\alpha$-approximation to the optimal welfare.  
For some $\epsilon>0$, if both $\upr_{v_2,\alpha,\epsilon,S}$ and $\lwr_{v_2,\alpha,\epsilon,S}$ belong to $\V_1$, then for every valuation $v_2\in \V_2$ and every bundle $S\subseteq [m]$, 
$$
\frac{1}{\alpha}\cdot v_2([m])- v_2([m]\setminus S)+\epsilon 
\ge \Menu_1^{\mathcal{A}}(v_2)(S)-\Menu_1^{\mathcal{A}}(v_2)(\emptyset)\ge \alpha\cdot v_2([m])- v_2([m]\setminus S)-\epsilon 
$$
\end{proposition}
\begin{proof}
Denote with $v_1^{\up}$ and $v_1^{\low}$ the valuations $\upr_{v_2,\alpha,\epsilon,S}$ and $\lwr_{v_2,\alpha,\epsilon,S}$ respectively.    
Consider the possible allocations that achieve an $\alpha$ approximation given the valuation profile $(v_1^{\up}, v_2)$.

Consider a bundle $T$ that contains the desired bundle $S$. Using the fact that $v_1^{\up}$ is single-minded, we have that
$$v_1^{\up}(T) + v_2([m]\setminus T) \ge v_1^{\up}(S) + v_2([m]\setminus S) = \frac 1 \alpha v_2([m]) + \epsilon.$$
On the other hand, if a bundle $T$ does not contain $S$, then: $$v_1^{\up}(T) + v_2([m] \setminus T) \le v_2([m]) < \alpha \cdot \left(\frac 1 \alpha v_2([m]) + \epsilon\right).$$
Thus, to get an $\alpha$-approximation to the optimal welfare, we must allocate a bundle containing $S$ to bidder $1$ given the valuation profile $(v_1^\up,v_2)$. 

Analogously, consider the possible allocations that achieve an $\alpha$ approximation  given $(v_1^{\low}, v_2)$. 
We have $v_1^{\low}(S) + v_2([m]\setminus S) = \alpha \cdot v_2([m]) - \epsilon < \alpha\cdot v_2([m])$, which implies that in this case any allocation awarding $S$ to bidder $1$ cannot possibly be an $\alpha$-approximation to the optimal welfare.

We now achieve the desired bounds on payments by using the fact that the mechanism is truthful, so given $v_1^{\up}$, bidder $1$ will not want to misreport the valuation  $v_1^{\low}$ and vice-versa.
Let $X^{\up}$ be the bundle that bidder $1$ wins given $(v_1^\up,v_2)$.
Using truthfulness and the monotonicity of the menu (\autoref{prop-menu-mono}), we have that:
\begin{equation*}
    v_1^{\up}(S)-\Menu_1^{\mathcal{A}}(v_2)(S)
    \ge v_1^{\up}(X^{\up})-\Menu_1^{\mathcal{A}}(v_2)(X^{\up})
    \ge v_1^{\up}(\emptyset)-\Menu_1^{\mathcal{A}}(v_2)(\emptyset)
\end{equation*}
Rearranging this inequality and applying the definition of $v_1^{\up}$:
\begin{equation*}
\Menu_1^{\mathcal{A}}(v_2)(S)-\Menu_1^{\mathcal{A}}(v_2)(\emptyset)
\le v_1^{\up}(S) - v_1^{\up}(\emptyset)
= \frac{1}{\alpha}\cdot v_2([m])-v_2([m]\setminus S)+\epsilon.
\end{equation*}
Thus, we obtain the first part of the lemma. 

For the second part, let  $X^{\low}$ be the bundle that bidder $1$ wins given $(v_1^\low,v_2)$. 
Since $v_1^{\low}(X^{\low})=0 = v_1^\low(\emptyset)$,
applying truthfulness and menu monotonicity gives: 
$$v_1^{\low}(\emptyset)-\Menu_1^{\mathcal{A}}(v_2)(\emptyset) \ge v_1^{\low}(X^{\low})-\Menu_1^{\mathcal{A}}(v_2)(X^{\low})\ge v_1^{\low}(S)-\Menu_1^{\mathcal{A}}(v_2)(S),$$
and again rearranging, we have
$$\Menu_1^{\mathcal{A}}(v_2)(S)-\Menu_1^{\mathcal{A}}(v_2)(\emptyset) \ge v_1^{\low}(S)-v_1^{\low}(\emptyset) = \alpha\cdot v_2([m])-v_2([m]\setminus S)-\epsilon.$$
which finishes the proof.
\end{proof}

\section{An Example of a $2$-Width-Family}
\label{app:example}
In this section, we give an example of a $2$-width-family, to demystify our construction in \autoref{sec:real-proof}. 
Consider the case where there are $6$ items, which we denote with $\{1,2,3,4,5,6\}$. 
Consider the set family defined as follows: 
\begin{align*}
 \mathcal{G}&=\big\{G_1=\{1,2\},G_2=\{3,5,6\}\big\} \\
 \mathcal H_1^{(0)}&= \big\{H_{1,1}^{(0)}=\{3,4\},H_{1,2}^{(0)}=\{5,6\}\big\} \\
  \mathcal H_1^{(1)}&= \big\{H_{1,1}^{(1)}=\{1,2\},H_{1,2}^{(1)}=\{1\}\big\} \\
   \mathcal H_2^{(0)}&= \big\{H_{2,1}^{(0)}=\{2,4\},H_{2,2}^{(0)}=\{1\}\big\} \\
  \mathcal H_2^{(1)}&= \big\{H_{2,1}^{(1)}=\{3,5\},H_{2,2}^{(1)}=\{5\}\big\} 
\end{align*}
Now, we define:
\begin{equation*}
    \vb b=(0,1), \quad  \vb {C} = \begin{pmatrix}
0 & 1 \\
0 & 1
\end{pmatrix}
\end{equation*}
Thus, the collection $\mathcal F[\vb b,\vb C]$
is composed of the subsets $S_{1,1},S_{1,2},S_{2,1},S_{2,2}$ that are defined as follows:
\begin{align*}
    b[1]=0,C[1][1]=0 \implies S_{1,1}&= G_1\cup H_{1,1}^{(0)}= \{1,2,3,4\} \\
        b[1]=0,C[1][2]=1 \implies S_{1,2}&= G_1\cup (\overline G_1\setminus H_{1,2}^{(0)}) = \{1,2,3,4\} \\                b[2]=1,C[2][1]=0 \implies S_{2,1}&= \overline G_2\cup  H_{2,1}^{(1)} = \{1,2,3,4,5\}   \\
        b[2]=1,C[2][2]=1 \implies S_{2,2}&= \overline G_2\cup  (G_2\setminus H_{2,2}^{(1)})  = \{1,2,3,4,6\}
\end{align*}
Observe that by definition $\mathcal F[\vb b,\vb C]$ is $2$-sparse, because there is no one subset among $S_{1,1},S_{1,2},S_{2,1},S_{2,2}$  that covers $\{1,2,3,4,5,6\}$. However, it is not $3$-sparse, because $S_{2,1}\cup S_{2,2}=\{1,2,3,4,5,6\}$, which implies that the family $\mathcal F$ is not $3$-independent.

\section{On The Precision of Mechanisms and Valuations}
\label{app:uniform}
In this section, we dive deeper into an assumption of the taxation framework; in particular, we remark on the precision of the bidders' valuations vs. the mechanism's prices. 
Observe that when discussing the communication complexity of combinatorial auctions, it is necessary to assume that the valuation functions of each bidder has values that are bounded in some range. 
The reason for it is that otherwise, even the communication complexity of a single item auction is infinite, regardless of incentives.\footnote{
    Consider a single-item auction with two bidders with values in the infinite range $R=\{0,\ell,\ell^2,\ell^3,\ldots\}$ for the item, where $\ell$ is arbitrarily large. 
    Applying the fooling set argument for the pairs $(x,x)_{x\in R}$ gives that the communication complexity of every mechanism that gives an approximation better than $\ell$ for the optimal welfare requires at least one transcript for every element of $R$, so it requires infinite communication.
} 
Therefore, when formally considering a class of valuation functions, it is necessary to specify not only their properties, i.e., whether they are subadditive or single-minded, but also their precision. 

Formally handling the precision of valuations has been discussed before, e.g. in \cite[Appendix A.4.1]{D16b}.
As discussed in our \autoref{fn:precision-preliminaries}, the basic approach is to assume all valuations can be represented by $k$-bit numbers for some implicit parameter $k$, and to formally define a protocol for some valuation class as a \emph{family} of protocols $(P_k)_k$ for different precision $k$ (with the added assumption that when $k' > k$ and $P_{k'}$ is run on $k$-bit numbers, the output is the same as the output of $P_k$).
To our knowledge, all prior upper and lower bounds %
in the algorithmic mechanism design literature hold in this formal model while implicitly hiding factors of $\poly(k)$.

We will now address an additional assumption of the taxation framework, and show that it is in fact not necessary. 
Namely, \cite{D16b} assumes that mechanisms that are ``{precision-aligned}''. A mechanism $\mathcal A:\mathcal V \to \Sigma \times \mathbb R^n$ is \emph{precision-aligned} if the precision of the mechanism is equal to the precision of $\mathcal V$.
More concretely, if every valuation in $\mathcal V$ has values in a certain range $R$, then the mechanism outputs prices that are in range $R$.

However, one can easily show that, in fact, assuming precision-aligned mechanisms is without loss of generality.
Thus, the results of \cite{D16b} (and hence, other results which build on these, such as ours and those of \cite{AKSW20}) in fact hold for all mechanisms, not just for precision-aligned ones. 
Formally:
\begin{proposition}[Communicated by \cite{Dob24}]
    Let $\mathcal A:\mathcal V\to \Sigma \times \mathbb R^n$ be a truthful mechanism.   Then, there exists a truthful precision-aligned mechanism $\mathcal A'$ such that:
    \begin{enumerate}
        \item For every valuation profile $v\in \mathcal V$, the mechanisms $\mathcal A$ and $\mathcal A'$ output the same allocation.
        \item The communication complexity of $\mathcal A'$ is the same as the communication complexity of $\mathcal A$.
    \end{enumerate}
\end{proposition}
\begin{proof}
    We assume for simplicity that the range of the valuations in 
$\mathcal V$ is integer values in $\{0,\ldots, H\}$, though an analogous proof works for every fixed decimal precision.  

The mechanism of $\mathcal A'$ that we propose is equivalent to the mechanism $\mathcal A$, except that the prices in each leaf are rounded down to the nearest integer. Thus, we get that the mechanisms $\mathcal A'$ and $\mathcal A$ have the same precision and the same communication complexity. It remains to show that $\mathcal A'$ is truthful. We remind that by the taxation principle, it suffices to show that
for every player $i$, and every valuation profile $(v_{i},v_{-i})$, the bundle that $\mathcal A'$ allocates to bidder $i$ remains the most profitable given the new rounded prices. 

Fix such player $i$ and a valuation profile $(v_i,v_{-i})$.
Let $S$ be the most 
profitable bundle, and let $T$ be some other bundle. Denote their prices given the (old) mechanism $\mathcal A$ with $p_S$ and $p_T$, respectively. 
Since $\mathcal A$ is truthful, we have that $v_i(S)-p_S\ge v_i(T)-p_T$, and our goal is to show that:
$v_i(S)-\lfloor p_S\rfloor \ge v_i(T)-\lfloor p_T\rfloor$.

To prove this, we represent $p_S$ as an integer $n_S$ and a fraction $f_S$, where $p_S=n_S+f_S$, where $n_S=\lfloor p_S \rfloor$ and $f_S \in [0,1)$. 
We define  $n_T$ and $f_T$ analogously. Therefore, we have that:
\begin{equation}\label{eq-precision}
v_i(S)-n_S-f_S \ge v_i(T)-n_T-f_T     
\end{equation}
Thus, we now have that the integer part of the left side is $v_i(S)-n_S$ and the integer part of the right side is $v_i(T)-n_T$. Since both  $f_S,f_T\in [0,1)$, then their difference satisfies that $|f_T-f_S|<1$, so for \autoref{eq-precision} to hold, it has to be the case that $v_i(S)-n_S \ge v_i(T)-n_T$. 
We remind that $n_S=\lfloor p_S \rfloor$ and $n_T=\lfloor p_T \rfloor$, so we have that   $v_i(S)-\lfloor p_S \rfloor \ge v_i(T)-\lfloor p_T \rfloor$.
Moreover, $n_S$ and $n_T$ are in fact the menu prices in the new mechanism $\mathcal{A}'$, so this mechanism is in fact truthful, as needed.
\end{proof}

\end{document}